\newif\ifdraft
\draftfalse

\pdfoutput=1

\newif\iftechreport
\techreporttrue

\documentclass[11pt,twocolumn]{article}
\usepackage[margin=0.8in]{geometry}

\usepackage{tikz}
\usetikzlibrary{fit}
\usetikzlibrary{shapes}
\usetikzlibrary{positioning}
\usetikzlibrary{arrows}
\usetikzlibrary{calc}
\usetikzlibrary{decorations.text}
\usetikzlibrary{decorations.pathmorphing}
\usetikzlibrary{decorations.pathreplacing}
\usetikzlibrary{automata}
\usetikzlibrary{shapes.callouts}
\usepackage[framemethod=tikz]{mdframed}

\usepackage[hyphens]{url}

\usepackage[users=SPGAM, \ifdraft \else hide \fi]{uchanges}
\ChangesSetUserColor{G}{blue}
\ChangesSetUserColor{M}{red}
\ChangesSetUserColor{A}{purple}
\ChangesSetUserColor{P}{teal}
\ChangesSetUserColor{S}{orange}

\usepackage{xspace,balance}
\usepackage{amssymb}
\usepackage{mathtools}
\usepackage{hyperref}
\usepackage[noabbrev, capitalise]{cleveref}
\crefname{algocf}{prot.}{prots.}
\Crefname{algocf}{Protocol}{Protocols}
\usepackage[font={small}]{caption} 
\usepackage[ruled,vlined,linesnumbered]{algorithm2e}
\usepackage{listings}
\SetKwInOut{Input}{input}
\SetKwInOut{Output}{output}
\SetEndCharOfAlgoLine{}
\usepackage{paralist}
\usepackage{array}

\usepackage{multicol}
\usepackage{balance}
\usepackage{wasysym}

\usepackage[operators,primitives,sets,keys,ff,adversary]{cryptocode}

\crefname{lemma}{lemma}{lemmas}

\newcommand{\sysname}{Fulgor\xspace}
\newcommand{\sysnamenonblocking}{Rayo\xspace}

\newcommand{\ttp}{\ensuremath{\mathcal{F}}\xspace}
\newcommand{\blockchain}{\ensuremath{\mathsf{B}}\xspace}
\newcommand{\readblock}{\ensuremath{\mathsf{read}}\xspace}
\newcommand{\openchannel}{\ensuremath{\mathsf{open}}\xspace}
\newcommand{\closechannel}{\ensuremath{\mathsf{close}}\xspace}
\newcommand{\success}{\ensuremath{\mathsf{success}}\xspace}

\newcommand{\VV}{\ensuremath{\mathbb{V}}\xspace}
\newcommand{\EE}{\ensuremath{\mathbb{E}}\xspace}

\newcommand{\capacity}{capacity\xspace}

\newcommand{\capa}{\ensuremath{\mathsf{cap}}\xspace}

\newcommand{\acc}{address\xspace}
\newcommand{\accs}{addresses\xspace}

\newcommand{\pcn}{PCN\xspace}
\newcommand{\pcns}{PCNs\xspace}

\newcommand{\pc}{payment channel\xspace}
\newcommand{\pcs}{payment channels\xspace}

\newcommand{\pcstate}{\ensuremath{\textit{channel-state}}\xspace}
\newcommand{\ops}{\ensuremath{\textit{decision}}\xspace}

\newcommand{\cur}{\ensuremath{\textit{cur}}\xspace}
\newcommand{\q}{\ensuremath{\textit{Q}}\xspace}

\newcommand{\ggroup}{\ensuremath{\mathcal{G}}\xspace}

\newcommand{\df}{\ensuremath{:=}\xspace}

\newcommand{\add}{\ensuremath{\textit{append}}\xspace}

\newcommand{\delete}{\ensuremath{\textit{delete}}\xspace}

\newcommand{\agreef}{\ensuremath{f}\xspace}

\newcommand{\send}{\ensuremath{\mathsf{Send}}\xspace}

\newcommand{\hold}{\ensuremath{\mathsf{forward}}\xspace}

\newcommand{\release}{\ensuremath{\mathsf{abort}}\xspace}
\newcommand{\confirm}{\ensuremath{\mathsf{accept}}\xspace}

\newcommand{\valtext}{\capa}

\newcommand{\offchain}{off-chain\xspace}
\newcommand{\onchain}{on-chain\xspace}

\newcommand{\uid}{\ensuremath{u}\xspace}

\newcommand{\user}{user\xspace}
\newcommand{\users}{users\xspace}

\newcommand{\Users}{Users\xspace}

\newcommand{\nw}{\textit{nw}\xspace}

\newcommand{\payment}{payment\xspace}
\newcommand{\payments}{payments\xspace}
\newcommand{\Payment}{Payment\xspace}
\newcommand{\Payments}{Payments\xspace}

\newcommand{\payer}{sender\xspace}
\newcommand{\payee}{receiver\xspace}

\newcommand{\pparagraph}[1]{ \smallskip \noindent \textbf{#1.\ }}

\newcommand{\sender}{\ensuremath{\mathsf{Sdr}}\xspace}
\newcommand{\receiver}{\ensuremath{\mathsf{Rvr}}\xspace}

\newcommand{\Txid}{\ensuremath{\mathsf{Txid}}\xspace}

\newcommand{\bit}[1]{\{0,1\}^{#1}}

\newcommand{\Setup}{\ensuremath{\mathsf{Setup}}}
\newcommand{\prob}[1]{\ensuremath{\mathsf{Pr}\left[#1\right]}\xspace}

\newcommand{\negl}{\ensuremath{\mathsf{negl}}}

\newcommand{\ppt}{\ensuremath{\mathsf{PPT}}\xspace}

\newcommand{\pay}{\ensuremath{\mathsf{pay}}\xspace}
\newcommand{\test}{\ensuremath{\mathsf{test}}\xspace}

\newcommand{\consensustwo}{\ensuremath{\mathsf{2ProcCons}}\xspace}
\newcommand{\openc}{\ensuremath{\mathsf{openChannel}}\xspace}
\newcommand{\closec}{\ensuremath{\mathsf{closeChannel}}\xspace}
\newcommand{\val}{\ensuremath{\textit{v}}\xspace}
\newcommand{\simulate}{\ensuremath{\mathcal{S}}\xspace}
\newcommand{\environment}{\ensuremath{\mathcal{E}}\xspace}

\newcommand{\relation}{\ensuremath{\mathcal{R}}\xspace}
\newcommand{\lang}{\ensuremath{\mathcal{L}}\xspace}

\newcommand{\newhtlc}{Multi-Hop HTLC\xspace}
\newcommand{\newhtlclong}{Multi-Hop Hash Time-Lock Contract\xspace}

\newcommand{\money}{coins\xspace}

\newcommand{\hl}[1]{\color{cyan}#1\color{black}}

\newcommand{\timeout}{\ensuremath{t}\xspace}
\newcommand{\fee}{\ensuremath{f}\xspace}
\newcommand{\cid}{\ensuremath{c}}
\newcommand{\cstate}{\ensuremath{h}\xspace}
\newcommand{\pair}[2]{\ensuremath{\cid_{\langle{\ensuremath{#1}},{\ensuremath{#2}}\rangle}}\xspace}
\newcommand{\colorchange}{light blue\xspace}

\newcommand{\queue}{\ensuremath{\mathcal{W}\xspace}}
\newcommand{\queued}{\ensuremath{\textsf{queued}}\xspace}

\newcommand{\correctbal}{balance security\xspace}
\newcommand{\Correctbal}{Balance security\xspace}

\usepackage{abstract}
\usepackage{amsthm}
\newtheorem{definition}{Definition}
\newtheorem{lemma}{Lemma}
\newtheorem{theorem}{Theorem}
\usepackage{float}
\newcommand{\mytitle}{Concurrency and Privacy with Payment-Channel Networks\thanks{This is the 
revision 6 September 2017. The most recent version is available at \url{https://eprint.iacr.org/2017/820} }} 
\date{}

\begin{document}

\title{\mytitle} 

\author{Giulio Malavolta\thanks{Both authors contributed equally and are considered to be co-first authors.}\\
Friedrich-Alexander-University Erlangen-N\"urnberg\\
malavolta@cs.fau.de\\
\and
Pedro Moreno-Sanchez$^\dagger$\\
Purdue University\\
pmorenos@purdue.edu\\
\and
Aniket Kate\\
Purdue University\\
aniket@purdue.edu\\
\and
Matteo Maffei\\
TU Wien\\
matteo.maffei@tuwien.ac.at\\
\and
Srivatsan Ravi\\
University of Southern California\\
srivatsr@usc.edu\\
}
%

\twocolumn[
  \begin{@twocolumnfalse}
	\maketitle


\begin{abstract}
\REPLACEAfor{SGM}{170613}{Blockchain payment protocols based on global consensus}{Permissionless blockchains protocols} such as Bitcoin are inherently limited in transaction throughput and 
latency. Current efforts to address this key issue focus on off-chain \payment channels that can be combined in a Payment-Channel Network
(\pcn) to enable an unlimited 
number of \payments without requiring to access the blockchain other than 
to register the initial and final \capacity of each channel. While this approach 
paves the way for low latency and high throughput of \payments, its deployment in 
practice raises several privacy concerns 
as well as  technical challenges related to the inherently  concurrent nature of \payments
\REMOVEA{, such as race conditions and deadlocks,} that have not been sufficiently studied so far.

 In this work, 
 we lay the foundations for privacy and concurrency in \pcns, presenting a formal definition in the Universal Composability framework 
 as well as practical and provably secure  solutions. In particular, we present \sysname and \sysnamenonblocking. 
 \sysname is the first payment protocol for \pcns that provides provable privacy guarantees for \pcns and is fully compatible with \REPLACEAfor{SGM}{170613}{current Bitcoin}{the Bitcoin scripting system}. 
 However, \sysname is a blocking protocol and therefore prone to deadlocks of concurrent \payments as in currently available \pcns. Instead, \sysnamenonblocking 
 is the first protocol for \pcns that enforces \emph{non-blocking progress} (i.e., at least one of the concurrent \payments terminates).  We show through a new impossibility result that \REPLACEPfor{SGAM}{170823}{the latter property}{non-blocking progress} necessarily comes
 at the cost of \REPLACEAfor{SGM}{170613}{breaking anonymity}{weaker privacy}. 
At the core of \sysname and \sysnamenonblocking is  \newhtlc, a new smart contract,  compatible with the Bitcoin scripting system, that provides  
 conditional payments while reducing running time and communication overhead with respect to previous approaches. 
 Our performance evaluation of \sysname and \sysnamenonblocking shows that  a \payment with $10$ intermediate \users takes  
 as few as 5 seconds\REMOVEAfor{SGM}{170613}{ and requires to communicate $17$ MB}, thereby demonstrating their feasibility to be deployed in practice.

\end{abstract}
  \end{@twocolumnfalse}
]

\maketitle


\section{Introduction}
\label{sec:intro}

Bitcoin\REMOVEPfor{SGA}{170519}{\footnote{Following the notation established in the Bitcoin community, we use  Bitcoin to refer to the Bitcoin network, 
bitcoin to denote the coins and BTC to denote the currency.}}~\cite{bitcoin} is a fully decentralized digital cryptocurrency network 
that is widely adopted today as an alternative monetary \payment system. 
Instead of accounting \payments in a ledger locally maintained by a trusted financial institute, 
these are logged in the Bitcoin blockchain, a database replicated 
among mutually distrusted \users around the world who update it by means of 
a  global consensus algorithm based on proof-of-work. Nevertheless, the permissionless nature 
of this consensus algorithm limits the transaction rate to tens of transactions per second 
whereas other \payment networks such as Visa support peaks of up to 
47,000 transactions per second~\cite{visa-peak}.

In the forethought of a growing number of Bitcoin \users and most importantly \payments about them, 
scalability is considered today an important concern among the Bitcoin 
community~\cite{scalability-article,bitcoin-scalability-faq}.
Several research and industry efforts are dedicated today to overcome this important 
burden~\cite{bitcoin-scalability-faq, flare, ln, dmc, bitcoin-scalability, bitcoin-contract}.

The use of Bitcoin \emph{payment channels}~\cite{paychannels, dmc} to realize \offchain \payments has flourished as a promising 
approach to overcome the Bitcoin scalability issue.
In a nutshell, a pair of \users open a payment channel by 
adding a single transaction to the blockchain where they lock their bitcoins in a deposit secured by a Bitcoin smart contract. 
Several  \offchain \payments can be then performed by locally agreeing on the new 
distribution of the deposit balance. 
Finally, the \users sharing the payment channel perform another Bitcoin transaction to add the final balances in the blockchain, 
effectively closing the payment channel.

In this manner,  
the blockchain is required 
to open and close a payment channel but not for any of the (possibly many) \payments between  users, thereby 
reducing the load on the blockchain and improving the transaction throughput. However, this simple approach is 
limited to direct \payments between two \users sharing an open channel. 
Interestingly,  it is in principle possible to leverage a path of opened payment channels from the sender 
to the receiver with enough \capacity to settle their \payments, effectively creating a \emph{payment-channel network (\pcn)}~\cite{ln}.

Many challenges must be overcome so that such a \pcn caters a wide deployment with a growing 
number of \users and \payments. In particular, today we know from similar \payment systems such as credit 
networks~\cite{ryan-network, trust-networks, ripple, stellar} 
that a fully-fledged \pcn must offer a solution to several issues, such as liquidity~\cite{liquidity-credit-networks, Moreno-SanchezM17}, 
network formation~\cite{formation-credit-networks}, routing scalability~\cite{bazaar, canal}, 
concurrency~\cite{silentwhispers}, and privacy~\cite{ripple-pets, privpay-ndss, silentwhispers} among others. 
\REPLACEPfor{SGA}{170216}{In this work, 
we shed light into some of these challenges in a \pcn. \TODOM{which of them have already been solved and by whom? and which ones do we address here? I would say privacy and concurrency, that's it, right? Anyway, we should be very precise about what is missing in the state-of-the-art and what we do here. }}
{

The Bitcoin community has started to identify these challenges~\cite{teechan, tumblebit, bolt, ln-scalability-challenges, Herrera-privacy, kristov-privacy-ln, bitcoin-scalability-faq, scalability-article}. Nevertheless, current \pcns are still immature\REMOVEAfor{SGM}{170613}{, only alpha implementations for testing purposes are available} and \REPLACEPfor{SGAM}{170823}{none of these challenges 
has been thoroughly studied yet}{these challenges require to be thoroughly studied}. In this work, we lay the foundations for privacy and concurrency in \pcns. Interestingly, we show that these two properties are connected to each other and that there exists an inherent trade-off between them.  
}



\pparagraph{The Privacy Challenge} 
It seems that payment channels necessarily improve the 
privacy of Bitcoin \payments as they are no longer logged in the blockchain. 
However,  such pervading 
idea has started to  be questioned by the community and it is not clear at this point 
whether a \pcn can offer sufficient privacy guarantees~\cite{ln-censor-resistance, kristov-privacy-ln, Herrera-privacy}.  
Recent research works~\cite{teechan, tumblebit, bolt} propose privacy preserving protocols for 
payment hub networks, where all \users perform \offchain \payments through a unique intermediary. Unfortunately, 
it is not clear how to extend these solutions to multi-hop \pcns. 
\REPLACEPfor{SGA}{170216}{Some recent works~\cite{flare,ln} aim at achieving privacy 
in \pcns through dedicated protocol, but the  
the lack of a rigorous  description of the privacy notions and threat models makes it infeasible to conduct a formal security analysis and to assess  the privacy guarantees actually provided by these proposals~\cite{kristov-privacy-ln}.  \MESSAGEMfor{SPGA}{170216}{weak, it would be better to point to specific limitations.}}
{

Currently, there exist some efforts in order to define a fully-fledged \pcn~\cite{npc-impl4, ln, raiden-nw, npc-impl3}. Among them, the Lightning Network~\cite{ln}
has emerged as  the most prominent \pcn among the Bitcoin community~\cite{ln-release}. 
However, its current operations do not provide all the privacy guarantees of interest in a \pcn. 
\MESSAGEPfor{SPGAM}{170519}{We might want to drop the following sentences as we no longer have routing protocol} 
For instance, the \REPLACEGfor{SAM}{170823}{calculation}{computation} of the 
maximum possible value to be routed through a \payment path requires that intermediate \users reveal the current \capacity of their \pcs 
to the sender~\cite[Section 3.6]{flare}, thereby leaking sensitive information. 
Additionally, the Bitcoin smart-contract used in the Lightning Network to enforce \REPLACEPfor{SGA}{170519}{that all \pcs in a path are correctly updated as a consequence 
of a \payment}{atomicity of updates for \pcs included in the \payment path},  requires to reveal a common 
hash value among each \user in the path   that can be used by intermediate \users to derive 
who is paying to whom~\cite{ln}.  As a matter of fact, 
while a plethora of academic papers have studied the privacy guarantees offered by current 
Bitcoin \payments on the Bitcoin blockchain~\cite{imc-meiklejohn, bitter-to-better, bitiodine, Koshy2014, 
	Reid2013, Androulaki2013, coinjoin-analysis}, there exists at present no  rigorous analysis of the privacy guarantees offered by or desirable in \pcns.  The lack of rigorous definitions for their protocols, threat model and  
 privacy notions, hinders a formal security and privacy analysis of ongoing attempts, let alone the development of provably secure and privacy-preserving solutions.}


\pparagraph{The Concurrency Challenge} The consensus algorithm, e.g., proof-of work in Bitcoin, eases the serialization of concurrent 
\onchain \payments. A miner with access to all concurrent \payments at a given time can easily serialize them following a set of predefined rules  
(e.g., sort them by \payment fee)  before they are added to the blockchain.  However, this is no longer the case in a \pcn:  
The bulk of \offchain \payments are not added to the blockchain \REMOVEAfor{SGM}{170613}{and therefore miners do not have a view of all concurrent \payments} and 
they cannot be serialized during consensus. 
Moreover, individual \users cannot avoid concurrency issues easily either as a \payment might involve several other \users apart from payer and payee. 

\NEWPfor{SGA}{170518}{In current \pcns such as the Lightning Network, a \payment is aborted as soon as a \pc in the path does not have enough capacity (possibly allocated 
for another in-flight \payment concurrently). This, however, leads to deadlock (and starvation) situations 
where none of the in-flight \payments terminates. }
\REPLACEPfor{SGAM}{170823}{A}{In summary, a}lthough concurrent \payments are likely to happen when current \pcns scale to a large number of \users and \offchain \payments, the 
inherent concurrency issues have not been thoroughly investigated yet.

\REMOVEPfor{SGA}{170216}{\MESSAGEPfor{SGAM}{170216}{Added more concrete arguments}\TODOM{what I find weak of this part is that it is not clear if existing proposals are great, just they lack a formal analysis, or they actually fall short of providing meaningful privacy guarantees.}}

\MESSAGEPfor{SGA}{170215}{We might want to discuss also here the ``concurrency challenge''. Currently is added to 
contribution and I will add it to the intro of concurrency section.}

\pparagraph{Our Contribution} This work makes the following contributions:

First, we formalize for the first time the security and privacy notions 
of interest for a \pcn, namely \emph{\correctbal},  \emph{value privacy} 
and \emph{sender/receiver anonymity}, following the universal composability (UC) framework~\cite{CanettiUC}.


Second, we study for the first time the concurrency issues in \pcns and present two protocols \sysname 
and \sysnamenonblocking that tackle this issue with different strategies. \sysname is a blocking protocol in 
line with concurrency solutions 
proposed in somewhat similar payment networks such as credit networks~\cite{silentwhispers, ripple} that 
can lead to deadlocks where none of the concurrent \payments go through. Overcoming this challenge, 
\sysnamenonblocking is the first protocol for \pcns guaranteeing \emph{non-blocking} progress~\cite{HS11-progress,AS85}. 
In doing so, \sysnamenonblocking ensures that 
at least one of the concurrent \payments terminates.

Third, we characterize an arguably surprising tradeoff between privacy and concurrency in \pcns. In particular, 
we demonstrate that any \pcn that enforces non-blocking progress inevitably reduces the 
anonymity set for sender and receiver of a \payment, thereby weakening the 
privacy guarantees.

Fourth, we formally describe the \newhtlclong (\newhtlc), a smart contract that lies at the core of \sysname and \sysnamenonblocking and which, 
in contrast to the Lightning Network, 
ensures  privacy properties even against 
\users in the payment path from payer to payee. 
\NEWPfor{SGAM}{170519}{
We formally define the \newhtlc contract and provide an 
efficient instantiation based on the recently proposed zero-knowledge proof system ZK-Boo~\cite{zkboo}, that improves on previous 
proposals~\cite{htlc-snarks} by reducing the  data required from $650$ MB to 
$17$ MB, the running time for the prover from $600$ ms to  $309$ ms and the running time for verifying from $500$ ms to $130$ ms.  Moreover, 
\newhtlc  does not require changes to the current Bitcoin \NEWPfor{SGAM}{170823}{scripting} system, \REPLACEPfor{SGAM}{170823}{and therefore 
can}{can thereby} be seamlessly deployed in current \pcns, and is thus  of independent interest.}


Finally, we have implemented a prototype of \sysname and \sysnamenonblocking in Python and evaluated the 
running time and communication cost to perform a  \payment. 
Our results show that a privacy-preserving \payment in a path with $10$ intermediate \users 
can be carried out in as few as $5$ seconds 
and incurs on $17$ MB of communication overhead.   
This shows that our protocols for \pcn are in line with with other privacy-preserving payment systems~\cite{privpay-ndss, silentwhispers}. 
Additionally, our evaluation shows that 
\sysname and \sysnamenonblocking can scale to cater a growing number of \users with a reasonably small overhead that can be 
further reduced with an optimized implementation.



\pparagraph{Organization}
\Cref{sec:background} overviews the required background. 
\Cref{sec:definition} defines the problem we tackle in this work and overviews \sysname and \sysnamenonblocking, our privacy preserving solution for \pcns. 
\Cref{sec:construction} details the \sysname protocol. 
\Cref{sec:concur-payments} describes our study of concurrency in \pcns and 
details the \sysnamenonblocking protocol. 
\Cref{sec:implementation} describes our implementation and the evaluation results. 
\Cref{sec:related-work} discusses the related work and~\cref{sec:conclusions} concludes 
this paper. 


\section{Background}
\label{sec:background}

In this section, we first overview the  notion of \pcs and we then describe payment-channel 
networks. 

\subsection{Payment Channels}
\label{sec:pay-channels}

A payment channel enables several Bitcoin \payments between two \users without committing every single 
\payment to the Bitcoin blockchain. The cornerstone of payment channels is depositing bitcoins into 
a multi-signature \acc controlled by both \users and having the guarantee that all bitcoins are refunded at a 
mutually agreed time if the channel expires. In the following, we overview the basics of payment channels 
and we refer the reader to~\cite{ln, dmc, towards-pn} for further details.

\begin{figure}[t]
\includegraphics[width=\columnwidth]{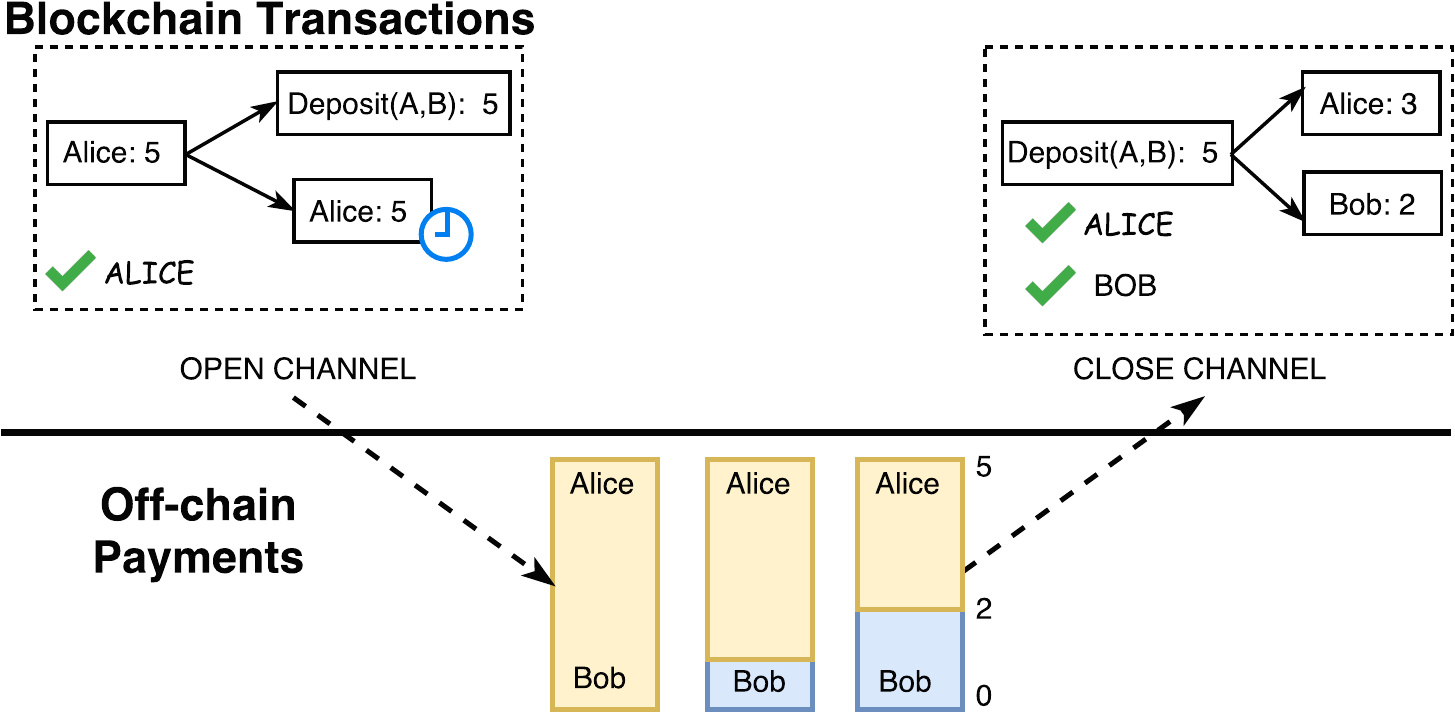}
\caption{{\bf{Illustrative example of \pc.}} White \NEWPfor{SGAM}{170823}{solid} boxes denote Bitcoin \accs and their current balance, \NEWPfor{SGAM}{170823}{dashed boxes represent Bitcoin transactions}, the clock denotes a time lock\NEWPfor{SGAM}{170823}{ contract~\cite{bitcoin-timelock}}, a \user name along a tick denotes  
her signature to validate the transaction and colored boxes denote the state of the \pc.  Dashed arrows denote temporal sequence. 
Alice first deposits $5$ bitcoins opening a \pc with Bob, then uses it to 
pay Bob \offchain. Finally, the \pc is closed with the most recent balance. \label{fig:paychannel}}
\end{figure}

In the illustrative example depicted in~\cref{fig:paychannel}, Alice opens a \pc with Bob with an initial capacity of $5$ bitcoins. 
This \emph{opening} transaction makes sure that Alice gets the money back after a certain timeout if the \pc is not used. 
Now, Alice can pay \offchain to Bob 
by adjusting the balance of the deposit in favor of Bob. Each \offchain \payment augments the balance for Bob and reduces it 
for Alice. When no more \offchain \payments are needed (or the capacity of the \pc is exhausted), the \pc is closed with a \emph{closing} 
transaction included in the blockchain. This transaction sends the deposited bitcoins to each \user according the most 
recent balance in the \pc.  

The \pc depicted in~\cref{fig:paychannel} is an example of \emph{unidirectional} channel: it can be used 
only \REPLACEPfor{SGAM}{170823}{to pay}{for payments} from Alice to Bob. \emph{Bidirectional} channels are defined to overcome this 
limitation as \offchain \payments in both directions are possible. 
Bidirectional payment channels operate in essence as the unidirectional version.\footnote{Technically, a bidirectional 
channel might require that both \users \REPLACEPfor{SGAM}{170823}{fund the deposit}{contribute funds to the deposit} in the opening transaction. However, current proposals~\cite{improve-ln} 
allow bidirectional channels with single deposit funder.  }
The major technical challenge consists in changing the 
direction of the channel. 
In \REPLACEPfor{SGAM}{170823}{our}{the running} example, assume that the current \pc balance $\textit{bal}$ is \{Alice: 4, Bob: 1\} and further assume that Bob 
pays \offchain one bitcoin back to Alice. The new \pc balance $\textit{bal}'$ is \{Alice: 5, Bob: 0\}. 
At this point, Alice benefits from $\textit{bal}'$ balance while Bob benefits from $\textit{bal}$. 
The solution to this discrepancy consists on that  
Bob and Alice make sure that 
any previous balance has been invalidated in favor of the most recent one.  
Different ``invalidation'' techniques have been proposed and we refer the reader to~\cite{ln, dmc, rusty-channel} 
for details. 

%
%
%

\REPLACEPfor{SGAM}{170823}{The current Bitcoin protocol needs to be slightly modified in order to fully support 
payment channels. In particular, transaction malleability~\cite{tx-malleability} must be fixed. 
Nevertheless, this 
fix, along with a set of other interesting new features in Bitcoin, is currently being discussed   
in the Bitcoin community as the Segwit soft-fork~\cite{segwit-benefits, bip101}. Therefore, payment channels 
are currently only fully supported in the Bitcoin testnet and  
the Bitcoin community expects to have them  available also in the main Bitcoin blockchain in the near 
future~\cite{segwit-signaling}. }
{The Bitcoin protocol has been updated recently to fully support \pcs. In particular, transaction 
malleability~\cite{tx-malleability}, along with a set of other interesting new features, have been added 
to the Bitcoin protocol with the recent adoption of Segregated Witness~\cite{segwit-signaling}. This 
event paves the way to the implementation and testing of \pcns on the main Bitcoin blockchain as 
of today~\cite{segwit-pcn-bitcoin-magazine}.
}

\subsection{A Payment Channel Network (PCN)}
\label{sec:nw-paychannels}

A \pcn can be represented as 
a directed graph $\GG = (\VV, \EE)$, where  the set \VV of vertices 
represents the Bitcoin accounts and the set \EE of weighted edges represents the payment 
channels. Every vertex $\uid \in \VV$ has associated a non-negative number that denotes the 
fee it charges for forwarding \payments. 
The weight on a directed edge $(\uid_1, \uid_2) \in \EE$ denotes the amount of 
remaining bitcoins that $\uid_1$ can pay to $\uid_2$. 
For ease of explanation, in the rest of the paper we represent a bidirectional channel 
between $\uid_1$ and $\uid_2$ as two directed edges, one in each 
direction.\footnote{In practice, there is a subtle difference: In a bidirectional channel between Alice and Bob, Bob 
can always return to Alice the bitcoins that she has already paid to him. However, 
if two unidirectional channels are used, Bob is limited to pay to Alice 
the \capacity of the edge Bob $\rightarrow$ Alice, independently of the bitcoins that he has received from 
Alice. 
Nevertheless, our simplification greatly ease the understanding of the rest of 
the paper and proposed algorithms \NEWGfor{SAM}{170823}{can be easily extended to} support bidirectional channels.}
Such a network can be used then to perform \offchain  \payments between two \users that 
do not have an open channel between them but  are connected by a path of open 
payment channels. 

\begin{figure}[t]
\includegraphics[width=\columnwidth]{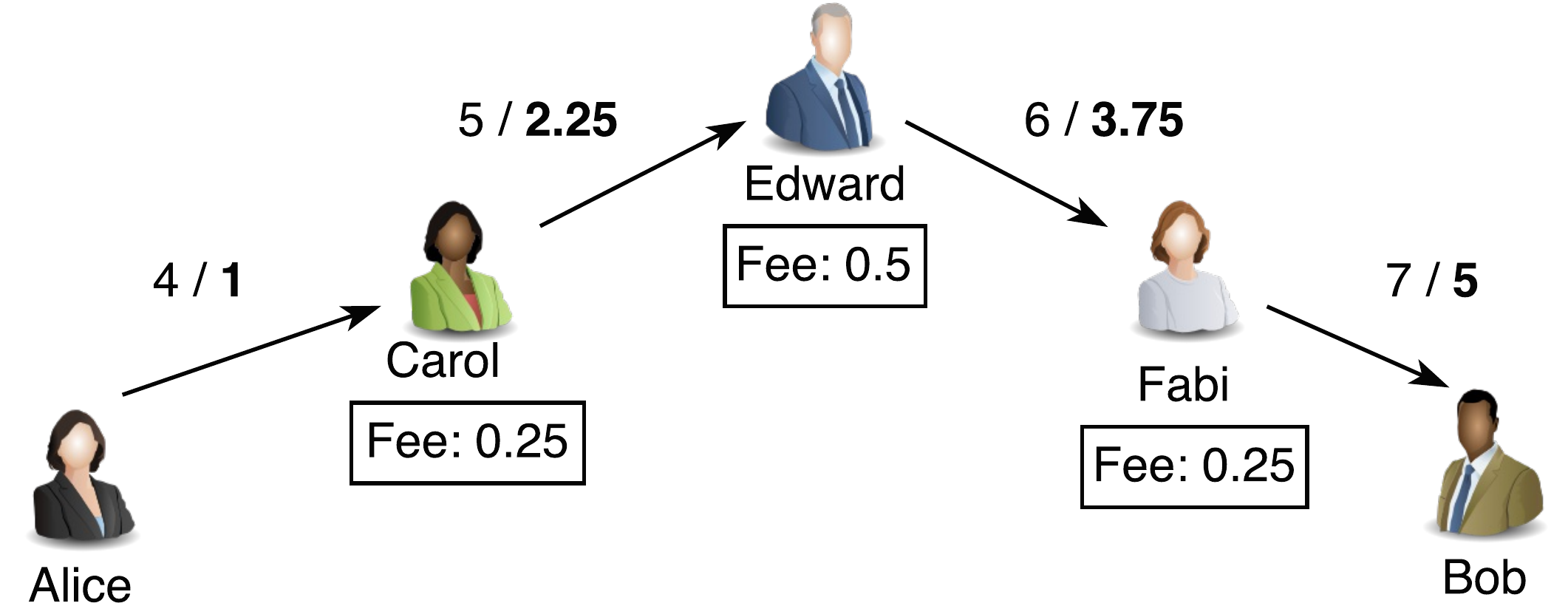}
\caption{{\bf{Illustrative example of a payment in a \pcn.}}
Non-bold (bold) numbers represent the \capacity of the channels before (after) the \payment 
from Alice to Bob. 
Alice wants to pay $2$ bitcoins to Bob via Carol, Edward and Fabi. Therefore, 
she starts the \payment with $3$ bitcoins (i.e., \payment amount plus fees). 
\label{fig:example-payment}
}
\end{figure}

The success of a \payment between two \users depends on the \capacity available along a path 
connecting the two \users and the fees charged by the \users in such path. 
Assume that $s$ wants to pay $\alpha$ bitcoins to $r$ and that they are connected 
through a path $s \rightarrow \uid_1 \rightarrow \ldots \rightarrow \uid_n \rightarrow r$. 
For their \payment to be successful, every link must have a \capacity $\gamma_i \geq \alpha'_i$, 
where $\alpha'_i = \alpha - \sum_{j=1}^{i-1}~\textit{fee}(\uid_j)$ (i.e., the initial \payment value 
minus the fees charged by intermediate \users in the path). 
At the end of a successful \payment, every edge in the path from $s$ to $r$ 
is decreased by $\alpha'_i$.  To ensure that $r$ receives exactly $\alpha$ bitcoins, $s$ must 
start the \payment with a value $\alpha^* = \alpha + \sum_{j=1}^{n}~\textit{fee}(u_j)$.

\REPLACEPfor{SGAM}{170823}{In~\cref{fig:example-payment}}{In the illustrative example of 
payment shown in~\cref{fig:example-payment}}, assume that Alice wants to pay 
Bob $2$ bitcoins. For that she needs to start a payment for a value of $3$ bitcoins ($2$ 
bitcoins plus $1$ bitcoin for the fees charged by \users in the path).  Then the \payment 
is settled as follows: capacity in the 
link Alice $\rightarrow$ Carol is reduced by $3$. Additionally, Carol charges 
\REPLACEPfor{SGAM}{170823}{$0.25$ bitcoin of fee}
{a fee if $0.25$ bitcoins} by 
reducing the \capacity of the link Carol $\rightarrow$ Edward by $2.75$ instead of $3$ bitcoins. 
Following the same reasoning, the link Edward $\rightarrow$ Fabi is set to \capacity $3.75$ and 
the link Fabi $\rightarrow$ Bob is set to $5$.

\REMOVEPfor{SGAM}{170519}{\TODOP{I don't know how important is this}\Payment fees play a key role in a \pcn. First, \payment fees motivate intermediate \users to 
forward a \payment from the payer to the payee as they get rewarded economically in the process. 
Second, since \payment fees are only collected for satisfactorily completed \payments, it is 
in the benefit of intermediate \users to comply with the protocol rules. Finally, \users in the \pcn 
might compete for getting more \payments routed through them by adapting their  fees, 
thereby creating a competition to maintain fees at a low value. }

\REMOVEMfor{SGA}{170216}{To enable \payments between arbitrary \users, the network must offer a good connectivity 
and should maintain a good liquidity, that is, the network must have a good set of paths 
with enough \capacity among their links.  However, since these networks are still at an early stage 
and no practical adoption has been seen yet, connectivity, liquidity or network formation 
remain interesting open challenges.  \TODOP{This paragraph is not really needed. It just 
adds possible interesting points for other works}}\MESSAGEMfor{SGA}{170216}{agreed, it can be discussed in the future work/conclusion section, if we really want to.}

\subsection{\REPLACEGfor{SAM}{170823}{Current}{State-of-the-Art} \pcns}
\label{sec:current-pcns-background}

The concepts of \pcs~\cite{tumblebit, teechan, dmc} and \pcns~\cite{towards-pn} have already attracted attention 
from the research community. 
In practice, there exist several ongoing implementations for a \REPLACEPfor{SGAM}{170823}{network of payment 
channels}{\pcn}  in Bitcoin~\cite{npc-impl1,npc-impl2,npc-impl3,npc-impl4}. Among them, the 
Lightning Network has emerged as  
the most prominent example in the Bitcoin community and an alpha implementation has 
been released recently~\cite{ln-release}. The idea of a \pcn has been 
proposed to improve scalability issues not only in Bitcoin, but also in other blockchain-based 
payment systems such as Ethereum~\cite{raiden-nw}.


\subsubsection{Routing in \pcns}
An important task in \pcns is to find paths with enough \capacity between 
\payer and \payee. \REPLACEGfor{SAM}{170823}{Assume that}{In our setting,} the network topology is known to every \user. This is \REPLACEGfor{SAM}{170823}{possible}{the case} since the 
opening of each payment channel is logged in the publicly available blockchain. Additionally, 
a gossip protocol between \users can be carried out to broadcast the existence of any 
payment channel~\cite{flare}. Further\REPLACEGfor{SAM}{170823}{ assume that}{more,} the fees charged by every \user \REPLACEGfor{SAM}{170823}{are}{can be} 
made public by similar means. 
Under these \REPLACEGfor{SAM}{170823}{assumptions}{conditions}, 
the \payer can locally 
calculate the paths between the \payer and the \payee.  
In the rest of the paper, we assume that the \payer chooses 
the path according to her own criteria. Nevertheless, 
we consider path selection as an interesting but orthogonal 
problem. 

%
%

\subsubsection{\Payments in \pcns}
\label{sec:atomic-payments}

\begin{figure}[b]
\includegraphics[width=\columnwidth]{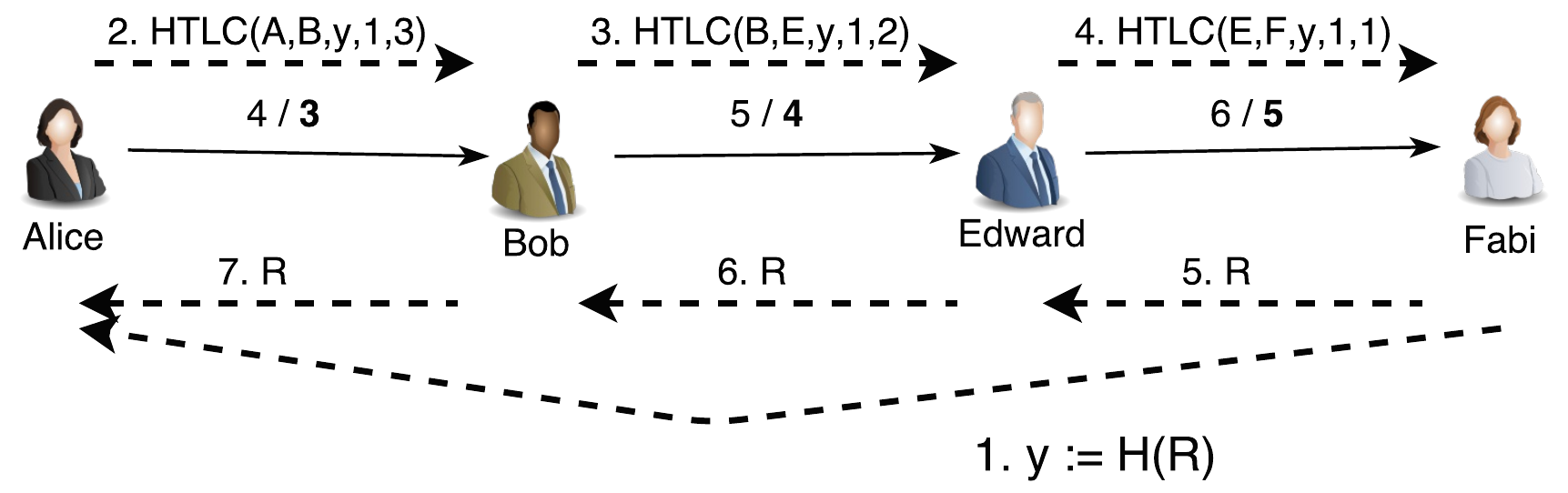}

\caption{{\bf Illustrative example of a \payment from Alice to Fabi for value $1$ using HTLC contract.} 
First, the condition is sent from Fabi to Alice. The condition is then forwarded among \users in the path to hold $1$ bitcoin at 
each \pc. Finally, the receiver shows $R$, releasing the held bitcoin at each \pc. For simplicity, 
we assume that there are no \payment fees in this example.\label{fig:htlc-payment}} 
\end{figure}

A \payment along a path of payment channels is carried out by updating the \capacity 
of each payment channel in the path according to the \payment amount and the 
associated fees (see~\cref{sec:nw-paychannels}). Such \NEWGfor{SAM}{170823}{an} operation rises the important challenge of \emph{atomicity}: 
either the \capacity of all channels in the path is updated or none of the channels 
is changed. Allowing changes in only some of the channels in the path can lead to the loss of bitcoins 
for a \user (e.g., a \user could pay certain bitcoins to the next \user in the path but 
never receive the corresponding bitcoins from the previous neighbor).

The current proposal in the Lightning Network consists of a smart contract called  \emph{Hash Time-Lock Contract} (HTLC)~\cite{ln}. 
This contract locks $x$ 
bitcoins that can be released only if the contract is fulfilled. The contract is defined, in terms 
of a hash value $y := H(R)$\NEWPfor{SGAM}{170824}{ where $R$ is chosen uniformly at random}, the amount of bitcoins $x$ and a timeout $t$, as follows:
\begin{center}
\resizebox{0.9\columnwidth}{!}{
\fbox{
\parbox{0.9\columnwidth}{
{\bf HTLC (Alice, Bob, $y$, $x$, $t$):}
\begin{asparaenum}
\item If Bob produces the condition $R^*$ such that $H(R^*) = y$ before $t$ days,\footnotemark Alice pays Bob $x$ bitcoins.
\item If $t$ days elapse, Alice gets back $x$ bitcoins.
\end{asparaenum}
}}}
\end{center}
\footnotetext{We use days here as in the original
description~\cite{ln}. Instead, recent proposals use \NEWPfor{SGAM}{170823}{the} sequence numbers of blocks \NEWPfor{SGAM}{170823}{as they appear in} the Bitcoin blockchain~\cite{bip68}.}


An illustrative example \NEWPfor{SGAM}{170823}{of the use of HTLC in a payment} is depicted in~\cref{fig:htlc-payment}. For simplicity, we assume that 
there are not \payment fees in this example. First, the payment 
amount (i.e., $1$ bitcoin) is set on hold from the sender to the receiver and then 
released from the receiver to the sender. In a bit more detail, 
after the receiver (Fabi) sends the condition \REMOVEPfor{SGA}{170216}{\TODOM{we should introduce this term before, in the HTLC algorithm}} to the sender (Alice), Alice sets an HTLC 
with her neighbor, effectively setting the payment value (i.e., $1$ bitcoin) on hold. 
Such HTLC is then set at each \pc in the path to the receiver. At this point, the receiver 
knows that the payment value is on hold at each \pc and thus she 
reveals the value $R$, that allows her to fulfill the contract and to settle the new \capacity 
at each \pc\NEWPfor{SGAM}{170823}{ in the path}.

It is important to note that every \user in the path sets the HTLC in
the outgoing \pc with a timeout smaller than the HTLC in the incoming
\pc. In this manner, the \user makes sure that she can pull bitcoins
from her predecessor after her bitcoins have been pulled from her
successor. An offline \user can outsource the monitoring of
fulfillments corresponding to open HTLC
contracts associated to her \pcs~\cite{channel-monitoring}. 

Although HTLC is fully compatible with Bitcoin, its use in practice 
leads to important privacy leaks: It is easy to see that the value of
the hash $H(R)$ uniquely identifies the users that took part in a
specific transaction. This fact has two main implications. First, 
any two colluding \users in a path can trivially derive the fact that 
they took part in the same \payment and this can be leveraged to reconstruct 
the identity of \payer and \payee.\footnote{
As noted in~\cite{bolt}, in a path $A \rightarrow I_1 \rightarrow I_2 \rightarrow I_3 \rightarrow B$, 
only $I_1$ and $I_3$ must collude to recover the identities of $A$ and $B$ as all the contracts in the 
path share the same $H(R)$.
} \MESSAGEGfor{SAM}{170823}{Change h to y in this section!}
Second, 
if the HTLC statements are uploaded
to the blockchain\NEWPfor{SGA}{170518}{ (e.g., due to uncollaborative intermediate \users in the payment path)}, an observer
can easily track the complete path used to 
route the \payment, even if she is not part of the \payment.  \REMOVEPfor{SGA}{170518}{\TODOM{it would be useful here to mention when HTLC statements are supposed to be uploaded on the blockchain, to complete the picture}}  In this work, we propose a novel \newhtlc smart contract that 
avoids this privacy problem while ensuring that no intermediate \user loses 
her bitcoins.  

\NEWPfor{SGA}{170518}{An important issue largely understudied in current \pcns is the handling of concurrent \payments that 
require a \emph{shared} \pc in their paths.
Current proposals simply abort a \payment if the balance at the shared \pc in the path is not enough. However, 
as we show in~\cref{sec:approach}, this approach can lead to a deadlock situation where none of simultaneous 
\payments terminates. We propose a \payment protocol that ensure non-blocking progress, that is, at 
least one of the concurrent \payments terminates. Moreover, we show an inherent tradeoff between concurrency and 
privacy for any fully distributed payment network.} 

\REMOVEPfor{SGA}{170518}{\TODOP{We might want to talk here about how they handle the concurrency and how we do it.}\TODOM{Rather than how we handle it, it would be useful to say here how concurrency is handled in current PCNs and which problems there exist (for privacy, we mention unlinkability, so we should mention some problem here for concurrency). Are there concrete facts to report?}}


%

\section{Problem Definition: Computation and Attacker model}
\label{sec:definition}
In this section, we first formalize a \pcn and underlying 
operations, and discuss the attacker model 
and our security and privacy goals. We then describe 
an ideal world functionality for our proposal, and present a system overview. \NEWG{Throughout the following description we implicitly assume that every algorithm takes as input the blockchain, which is publicly known to all users.} 

\begin{definition}[Payment Channel Network (\pcn)]
\label{def:nw-paychannels}
A \pcn is defined as \REPLACEPfor{SGAM}{170824}{$\nw := \GG(\VV, \EE)$}{graph 
$\GG := (\VV, \EE)$}, where $\VV$ is the set of 
Bitcoin accounts and $\EE$ is the set of currently open payment channels. 
A \pcn \NEWG{ is defined with respect to a blockchain $\blockchain$ and} is equipped 
with the three operations (\openc, \closec, \pay)  described below:

\begin{asparaitem}
\item $\openc(\uid_1, \uid_2, \beta, \timeout, \fee) \rightarrow \{1, 0\}.$ On input two Bitcoin \accs
$\uid_1, \uid_2 \in \VV$, an initial channel \capacity $\beta$, \NEWG{a timeout $\timeout$, and a fee value $\fee$,} if the operation 
is authorized by $\uid_1$, and $\uid_1$ owns at least $\beta$ bitcoins, \openc creates 
a new payment channel $(\pair{\uid_1}{\uid_2}, \beta, \fee, \timeout) \in \EE$, \NEWG{where $\pair{\uid_1}{\uid_2}$ is a fresh channel identifier. Then it uploads it to $\blockchain$ }\REMOVEG{, sets channel's \capacity to $\beta$, sets channel's fee to \fee, 
sets channel's timeout to \timeout, includes the channel in the underlying blockchain} and  
returns $1$. Otherwise, it returns $0$. 

\TODOP{Add a time to solve disputes before the channel is closed}
\item $\closec(\pair{\uid_1}{\uid_2},  \NEWG{\val}) \rightarrow \{1, 0\}.$  On input a \REPLACEG{\pc}{channel identifier}  
$\pair{\uid_1}{\uid_2}$ and \REPLACEG{its current}{a} balance \REPLACEG{\cstate}{$\val$} \NEWPfor{SGA}{170518}{ (i.e., the distribution of bitcoins locked in the channel between $\uid_1$ and $\uid_2$)}, 
if the operation is authorized by both $\uid_1$ and $\uid_2$, \closec removes the \NEWG{corresponding} channel from $\GG$, includes the balance 
\REPLACEG{\cstate}{$\val$} in \REPLACEG{the underlying blockchain}{$\blockchain$} and 
returns $1$. Otherwise, it returns $0$.  \REMOVEPfor{SGA}{170518}{\TODOM{perhaps we should say here what a balance \cstate is.}}


\item $\pay((\pair{s}{\uid_1}, \ldots, \pair{\uid_n}{r}), \val) \rightarrow \{1, 0\}.$ On input 
a list of \REPLACEG{\pcs}{channel identifiers} $ (\pair{s}{\uid_1}, \ldots, \pair{\uid_n}{r})$ and a \payment value $\val$, 
if the \pcs form a path from the sender ($s$) to the receiver ($r$) and each \pc $\pair{\uid_i}{\uid_{i+1}}$ in the path has 
at least a current balance   $\gamma_i \geq \val'_i$, 
where $\val'_i = \val - \sum_{j=1}^{i-1}~\textit{fee}(\uid_j)$, the \pay operation decreases the current balance for each 
\pc $\pair{\uid_i}{\uid_{i+1}}$ by $\val'_i$ and returns $1$. Otherwise, none of the balances at the \pcs 
is modified and the \pay operation returns $0$.

\REMOVEPfor{SGA}{170519}{Hereby, we say that a \payment  \emph{successfully terminates} when \pay returns $1$. 
We say that a \payment \emph{unsuccessfully terminates} when \pay returns $0$  and 
all the balances held at each \pc in the path are released. In any other case, we 
say that the \payment is \emph{aborted}. }

%

\end{asparaitem}
\end{definition}



\paragraph{Attacker model.} We consider a computationally efficient attacker that can shape the
network at her will by spawning users and corrupting an arbitrary
subset of them in an adaptive fashion. Once a user is corrupted, its
internal state is given to the attacker and all of the following
messages for that user are handed over to the attacker. \NEWMfor{SGA}{170518}{On the other hand, we assume that the communication between two non-compromised users sharing a \pc is confidential (e.g., through TLS). } Finally, the attacker
can send arbitrary messages on behalf of corrupted users. 
\REMOVEMfor{SGA}{170518}{Therefore, we consider a local active adversary that does not control the 
	communication among honest users and therefore cannot determine 
	whether two honest \users exchange a message. \TODOP{I am not sure 
		this is the correct way of writing the last sentence}}

\subsection{Security \REPLACEPfor{SGAM}{170823}{Goal}{and Privacy Goals}}
\label{sec:attacker-model}

\REMOVEPfor{SGA}{170519}{In the following, we characterize the security and privacy notions of interest in 
a \pcn. We defer their formal treatment to \cref{sec:ideal-world}. }
\NEWAfor{SG}{170519}{Against the above adversary, we identify the following security and privacy notions of interest:}
\begin{asparaitem}

\item {\bf \Correctbal.} Intuitively, \correctbal guarantees that 
any honest intermediate \user taking part in a \pay operation (as specified in Definition~\ref{def:nw-paychannels}) 
does not lose \money even when all other \users involved in  the \pay operation are corrupted.

\item {\bf Serializability.} We require that the executions of \pcn are \emph{serializable}~\cite{Pap79-serial}, 
i.e., for every concurrent execution of pay operations, 
there exists an \emph{equivalent} sequential execution. 
\item {\bf (Off-path) Value Privacy.}  Intuitively, value privacy
guarantees that for a \pay operation involving only honest \users, 
corrupted \users outside the payment path learn no information about the payment value.

\item {\bf (On-path) Relationship Anonymity~\cite{Pfitzmann08anonymity,Anoa-CSF}.} \NEWGfor{SAM}{170519}{
Relationship anonymity requires that, 
given two simultaneous successful \pay operations of the form $\set{\pay_i((\pair{s_i} {\uid_1}, \ldots, \pair{\uid_n}{r_i}), \val)}_{i \in [0,1]}$
with at least one honest intermediate \user $\uid_{j \in [1,n]}$,
corrupted intermediate \users cannot determine the pair $(s_i, r_i)$ for a given $\pay_i$ with 
probability better than $1/2$.}
\end{asparaitem}


\subsection{Ideal World Functionality}
\label{sec:ideal-world}

\REMOVEPfor{SGAM}{170519}{In this section, we first describe the ideal world functionality for a \pcn. We then discuss how this ideal world 
captures the security and privacy properties of interest in a \pcn. }

\begin{figure*}[tb]
\begin{mdframed}
{\bf Open channel:} On input $(\openchannel, \pair{\uid}{\uid'}, \val, \uid', \timeout, \fee)$ from a user $\uid$, the $\ttp$ checks whether $\pair{\uid}{\uid'}$ is well-formed (contains valid identifiers and it is not a duplicate) and eventually sends $(\pair{\uid}{\uid'}, \val, \timeout, \fee)$ to $\uid'$, who can either abort or authorize the operation. In the latter case, $\ttp$ appends the tuple $(\pair{\uid}{\uid'}, \val, \timeout, \fee)$ to $\blockchain$ and the tuple $(\pair{\uid}{\uid'}, \REPLACEPfor{SGAM}{170519}{0}{\val}, \timeout, \cstate)$ to $\mathcal{L}$, for some random $\cstate$. $\ttp$ returns $\cstate$ to $\uid$ and $\uid'$.

\smallskip
{\bf Close channel:} On input $(\closechannel, \pair{\uid}{\uid'}, \cstate)$ from a user $\in \{ \uid', \uid \}$ the ideal functionality $\ttp$ parses $\blockchain$ for an entry $(\pair{\uid}{\uid'}, \val, \timeout, \fee)$ and $\mathcal{L}$ for an entry $(\pair{\uid}{\uid'}, \val', \timeout', \cstate)$, for $\cstate \ne \bot$. If $\pair{\uid}{\uid'} \in \mathcal{C}$ \REMOVEGfor{SAM}{170519}{or $\pair{\uid}{\uid'}=\bot$} \NEWPfor{SGAM}{170519}{or $\timeout > |\blockchain|$ or $ \timeout' > |\blockchain|$} the functionality aborts. 
	 \REPLACEPfor{SGAM}{170519}{If $\timeout \ge |\blockchain|$ and $\timeout' \ge |\blockchain|$}{Otherwise,} $\ttp$ adds the entry\REMOVEPfor{SGAM}{170519}{s} $(\REMOVEPfor{SGAM}{170519}{\bot, }\pair{\uid}{\uid'}, \uid', \val', \timeout')$\REMOVEPfor{SGAM}{170519}{ and $(\bot, \pair{\uid}{\uid'}, \uid, (\val - \val'), \timeout')$} to $\blockchain$  and adds 
	$\pair{\uid}{\uid'}$ to $\mathcal{C}$. $\ttp$ then notifies both users involved with a message $(\pair{\uid}{\uid'}, \bot, \cstate)$.

\smallskip
{\bf Payment:} On input $(\pay, \val, 
	(\pair{\uid_0}{\uid_1}, \dots, \pair{\uid_n}{\uid_{n+1}}), (\timeout_0, \dots, \timeout_n))$ from a user $\uid_0$, $\ttp$ executes the following interactive protocol: 
	
	\begin{enumerate}
	\item  For all $i \in \{1, \dots, (n+1)\}$ $\ttp$ samples a random $\cstate_i$ and parses $\blockchain$ for an entry of the form 
		$(\pair{\uid_{i-1}}{\uid'_{i}}, \val_i, \timeout_i', \fee_i)$. If such an entry does exist $\ttp$ sends the tuple $(\cstate_i, \cstate_{i+1},$ 
		$\pair{\uid_{i-1}}{\uid_i}, \pair{\uid_i}{\uid_{i+1}}, \val \REPLACEPfor{SGAM}{170519}{+}{-} \sum^n_{j=i}\fee_j, \timeout_{i-1}, \timeout_i)$ to the user $\uid_i$ via an anonymous channel 
		(for the specific case of the receiver the tuple is only $(\cstate_{n+1}, \pair{\uid_{n}}{\uid_{n+1}} , \val, \timeout_n)$). Then $\ttp$ checks whether 
		for all entries of the form 
		$(\pair{\uid_{i-1}}{\uid_{i}}, \val_i', \cdot, \cdot) \in \mathcal{L}$ it holds that $\REPLACEPfor{SGAM}{170519}{(\val_i-\val_i')}{\val'_i} \ge \left(\val \REPLACEPfor{SGAM}{170519}{+}{-} \sum^n_{j=i}f_j\right)$ and that $\timeout_{i-1} \ge \timeout_i$. If this is the case $\ttp$ adds 
		$d_i = (\pair{\uid_{i-1}}{\uid_{i}}, (\val_i' \REPLACEPfor{SGAM}{170519}{+}{-} \NEWPfor{SGAM}{170519}{(}\val \REPLACEPfor{SGAM}{170519}{+}{-} \sum^n_{j=i}\fee_j\NEWPfor{SGAM}{170519}{)}), \timeout_i, \bot)$  to $\mathcal{L}$, where 
		$(\pair{\uid_{i-1}}{\uid_{i}}, \val_i', \cdot , \cdot) \in \mathcal{L}$ is the entry with the \REPLACEGfor{M}{170519}{highest}{lowest} $\val_i'$. If any of the conditions above is not met, $\ttp$ removes from $\mathcal{L}$ all the entries $d_i$ added in this phase and aborts. 
		 
		 \item For all $i \in \{(n+1), \dots, 1\}$ $\ttp$ queries all $\uid_i$ with $(\cstate_i, \cstate_{i+1})$, through an anonymous channel. Each user can reply with either $\top$ or $\bot$. Let $j$ be the index of the user that returns $\bot$ such that for all $i > j: \uid_i$ returned  $\top$. If no user returned $\bot$ we set $j = 0$.

		\item For all $i \in \{ j+1, \dots, n\}$ the ideal functionality $\ttp$ updates $d_i \in \mathcal{L}$ (defined as above) to $(-, -, -, \cstate_i)$ and notifies the user of the success of the operation with with some distinguished message 
		$(\success, \cstate_i, \cstate_{i+1})$. For all $i \in \{0, \dots, j\}$ (if $j \ne 0$) $\ttp$ removes $d_i$ from $\mathcal{L}$ and notifies the 
		user with the message $(\bot, \cstate_i, \cstate_{i+1})$.
	\end{enumerate}
\end{mdframed}
\caption{Ideal world functionality for \pcns.\label{fig:ideal-world}}
\end{figure*}

\pparagraph{Our Model} The \users of the network are modeled as interactive Turing machines that communicate with a trusted functionality $\ttp$ via secure and authenticated channels. We model the attacker $\adv$ as a probabilistic polynomial-time machine that is given additional interfaces to add users to the system and corrupt them. $\adv$ can query those interfaces adaptively and at any time. Upon corruption of a user $\uid$, the attacker is provided with the internal state of $\uid$ and the incoming and outgoing communication of $\uid$ is routed thorough $\adv$.

\pparagraph{Assumptions} We model anonymous communication between any two users of the network as an ideal functionality $\ttp_\mathsf{anon}$, as proposed in~\cite{CL05}. Furthermore, we assume the existence of a  blockchain $\blockchain$ that we model as a trusted append-only bulletin board (such as~\cite{bulletinboard}):  The corresponding ideal functionality $\ttp_\blockchain$ maintains $\blockchain$ locally and updates it according to the transactions between users. At any point in the execution, any user $\uid$ of the \pcn can send a distinguished message $\readblock$ to $\ttp_\blockchain$, who sends the whole transcript of $\blockchain$ to $\uid$. We denote the number of entries of $\blockchain$ by $|\blockchain|$. In our model, time corresponds to the number of entries of the blockchain $\blockchain$, i.e., time $t$ is whenever $|\blockchain| = t$. Our idealized process $\ttp$ uses $\ttp_\mathsf{anon}$ and $\ttp_\blockchain$ as subroutines, i.e., our protocol is specified in the $(\ttp_\mathsf{anon}, \ttp_\blockchain)$-hybrid model. \NEWGfor{SAM}{170823}{Note that our model for a blockchain is a coarse grained abstraction of the reality and that more accurate formalizations are known in the literature, see~\cite{hawk}. For ease of exposition we stick to this simplistic view, but one can easily extend our model to incorporate more sophisticated abstractions.}

\pparagraph{Notation} Payment channels in the Blockchain $\blockchain$ are of the form    
$(\pair{\uid}{\uid'}, \val, \timeout, \fee)$, where $\pair{\uid}{\uid'}$ is a unique channel identifier, $\val$ is the capacity of the channel, $t$ is the expiration time of the channel, and $f$ is the associated fee. For ease of notation we assume that the identifiers of the users $(\uid, \uid')$ are also encoded in $\pair{\uid}{\uid'}$. \NEWGfor{SAM}{170823}{We stress that any two users may have multiple channels open simultaneously.} The functionality maintains two additional internal lists $\mathcal{C}$ and $\mathcal{L}$. The former is used to keep track of the closed channels, while the latter records the off-chain payments. Entries in $\mathcal{L}$ are of the form $(\pair{\uid}{\uid'}, \val, t, h)$, where $\pair{\uid}{\uid'}$ is the corresponding channel, $\val$ is the amount of credit used, 
$\timeout$ is the expiration time of the payment, and $\cstate$ is the identifier for this entry.

\pparagraph{Operations} In~\cref{fig:ideal-world} we describe the interactions between $\ttp$ and the users of the \pcn. For simplicity, we only model unidirectional channels, 
although our functionality can be easily extended to support also bidirectional channels. The execution of our simulation starts with $\ttp$ initializing a pair of local empty lists $(\mathcal{L}, \mathcal{C})$. Users of a \pcn can query $\ttp$ to open channels and close them to any valid state in $\mathcal{L}$. On input a value $\val$ and a set of payment channels $(\pair{\uid_0}{\uid_1}, \dots, \pair{\uid_n}{\uid_{n+1}})$ from some user $\uid_0$, $\ttp$ checks whether the path has enough capacity (step 1) and initiates the payment. Each intermediate user can either allow the payment or \REPLACEPfor{SGAM}{170823}{not}{deny it}. Once \REPLACEPfor{SGAM}{170823}{it}{the payment} has reached the receiver, each user can again decide to interrupt the flow of the payment (step 2), i.e., pay instead of the sender. Finally $\ttp$ informs the involved nodes of the success of the operation (step 3) and adds the updated state to $\mathcal{L}$ for the corresponding channels.


\pparagraph{Discussion} Here, we show that our ideal functionality captures the security and privacy properties of interest for a \pcn. 

\begin{asparaitem}

\item \emph{\Correctbal.} 
\NEWG{Let $\uid_i$ be any intermediate hop in a payment $\pay((\pair{s}{\uid_1}, \ldots, \pair{\uid_n}{r}), \val)$. 
$\ttp$ locally updates in $\mathcal{L}$ the channels corresponding to the incoming and 
outgoing edges of $\uid_i$ such that the total balance of $\uid_i$ is increased by the \money she sets as a fee, unless the user actively prevents it (step 2).} Since $\ttp$ is trusted, \correctbal follows.

\item \emph{Serializability.} 
\NEWG{Consider for the moment only single-hop payments. 
It is easy to see that the ideal functionality executes them serially, i.e., any two concurrent payments can only happen on different links. 
Therefore one can trivially find a scheduler that performs the same operation in a serial order (i.e., in any order). 
By \correctbal, any payment can be represented as a set of atomic single-hop payments and thus serializability holds.}
 
\item \emph{Value Privacy.} In the ideal world, users that do not lie in the payment path are not contacted by $\ttp$ and therefore they learn nothing about the transacted value (for the off-chain payments). 
\item \emph{Relationship Anonymity.} 
\NEWGfor{SAM}{170519}{Let $\uid_i$ be an intermediate hop in a payment. In the interaction with the ideal functionality, $\uid_i$ is only provided with a unique identifier for each payment. In particular, such an identifier is completely independent from the identifiers of other users involved in the same payment. It follows that, as long as at least one honest user $\uid_i$ lies in a payment path, any two simultaneous payments over the same path for the same value $\val$ are indistinguishable to the eyes of the user $\uid_{i+1}$. This implies that any {\em proper} subset of corrupted intermediate hops, for any two successful concurrent payments traversing all of the corrupted nodes, cannot distinguish in which order an honest $\uid_i$ forwarded the payments. Therefore such a set of corrupted nodes cannot determine the correct sender-receiver pair with probability better than $1/2$.}
\end{asparaitem}

\pparagraph{UC-Security} Let $\mathsf{EXEC}_{\tau, \adv, \mathcal{E}}$ be the ensemble of the outputs of the environment
  $\mathcal{E}$ when interacting with the adversary $\adv$ and parties
  running the protocol $\tau$ (over the random coins of all the
  involved machines).

\begin{definition}[UC-Security]\label{def:UCsec}
A protocol $\tau$ UC-realizes an ideal functionality $\ttp$ if
for any adversary $\adv$ there exists a simulator $\mathcal{S}$ such that for
any environment $\mathcal{E}$ the ensembles $\mathsf{EXEC}_{\tau, \adv,
  \mathcal{E}}$ and $\mathsf{EXEC}_{\ttp, \mathcal{S}, \mathcal{E}}$ are computationally
indistinguishable.
\end{definition}
\pparagraph{Lower bound on byzantine users in \pcn} 
We observe that in \pcns that contain channels in which both the users are \emph{byzantine} (\`a la malicious)~\cite{LF82}, there is an inherent cost to concurrency. 
Specifically, in such a \pcn, if we are providing non-blocking progress, i.e., 
at least one of the concurrent \payments terminates, then it is impossible to provide serializability in \pcns (cf. Figure~\ref{fig:bottleneck}
in Appendix~\ref{sec:concurrency-proofs-lemmas}).
Thus, henceforth, all results and claims in this paper assume that in any \pcn execution, there does not exist a channel in which both its users are byzantine.
%
\begin{lemma}
\label{lemma:byzantine}
There does not exist any serializable protocol for the \pcn problem that provides non-blocking progress if there exists a \pc in which both \users
are byzantine.
\end{lemma}

%
%
%


%

\subsection{ Key Ideas and System Overview}
\label{sec:approach}

In the following, we give a high-level overview on how we achieve private and concurrent payments in \pcns.

\subsubsection{Payment Privacy}
The payment
operation must ensure the security and privacy properties of 
interest in a \pcn, namely \correctbal, value privacy and relationship 
anonymity. 
 A na\"ive approach towards achieving \correctbal would be to
use HTLC-based payments
(see~\cref{sec:atomic-payments}). This solution is however in inherent
conflict with anonymity: It is easy to see that contracts
belonging to the same transactions are \emph{linkable} among each
other, since they encode the same condition ($\cstate$) to release the payment. 
Our proposal, called \newhtlc, aims to remove this link among hops while
maintaining the full compatibility with the Bitcoin network. 

The idea
underlying \newhtlc is the following: At the beginning of an
$n$-hop transaction the sender samples $n$-many independent \REPLACEAfor{SGM}{170613}{points}{strings}
$(x_1, \dots, x_n)$. Then, for all $i \in 1, \ldots, n$, she sets
$y_i = H\left(\bigoplus^{n}_{j=i}x_j\right)$, where $H$ is an
arbitrary hash function. That is, each $y_i$ is the result of 
applying the function $H$ to all of the input values $x_j$ for
$j \ge i$ in an XOR combiner. The sender then provides the receiver
with $(y_n, x_n)$ and the $i$-th node with the tuple
$(y_{i+1},y_i, x_i)$. In order \REPLACEPfor{SGAM}{170823}{not to break}{to preserve} anonymity, the sender
communicates those values to the intermediate nodes over an anonymous channel. 
Starting from the
\emph{sender}, each pair of neighboring nodes $(u_{i+1},u_i)$ defines
a standard HTLC on inputs $(u_{i},u_{i+1}, y_i, b, t)$, where $b$ and
$t$ are the amount of bitcoin and the timeout parameter,
respectively. Note that the release conditions of the contracts are
uniformly distributed in the range of the function $H$ and therefore
the HTLCs of a single transaction are independent 
\REPLACEPfor{SGAM}{170823}{one another}{from each other}. Clearly, the mechanism described above works fine as long as
the sender chooses  each value \NEWPfor{SGAM}{170823}{$y_i$} according to the specification of the
protocol. We can enforce an honest behavior by including \REMOVEGfor{SAM}{170823}{recent
}non-interactive zero-knowledge proofs~\cite{zkboo}. \TODOP{I think we might 
need a figure for  this} \MESSAGEMfor{SPGA}{170518}{I agree}


\subsubsection{Concurrent \Payments}
It is possible that two (or more) simultaneous \payments share a \pc in their \payment 
paths in such a manner that none of the \payments goes through. 
In the example depicted in~\cref{fig:blocking-tx}, the \payment from Alice to Gabriel 
cannot be carried out as the \capacity in the \pc  between Fabi and Gabriel is already locked 
for the \payment from Bob to Edward. Moreover, this second \payment cannot be 
carried out either as the \capacity on the \pc between Carol and Edward is already 
locked. \NEWPfor{SGAM}{170518}{This deadlock situation is a generic problem of \pcns, where a \payment 
is aborted as soon as there exists a \pc in the path without enough capacity.}

 \REMOVEPfor{SGAM}{170518}{\TODOM{isn't this a general problem of PCNs? Then we should prominently say that (e.g., here and in the intro) since we are solving an existing problem (even for the non-private setting) }}

\begin{figure}[tb]
\includegraphics[width=0.9\columnwidth]{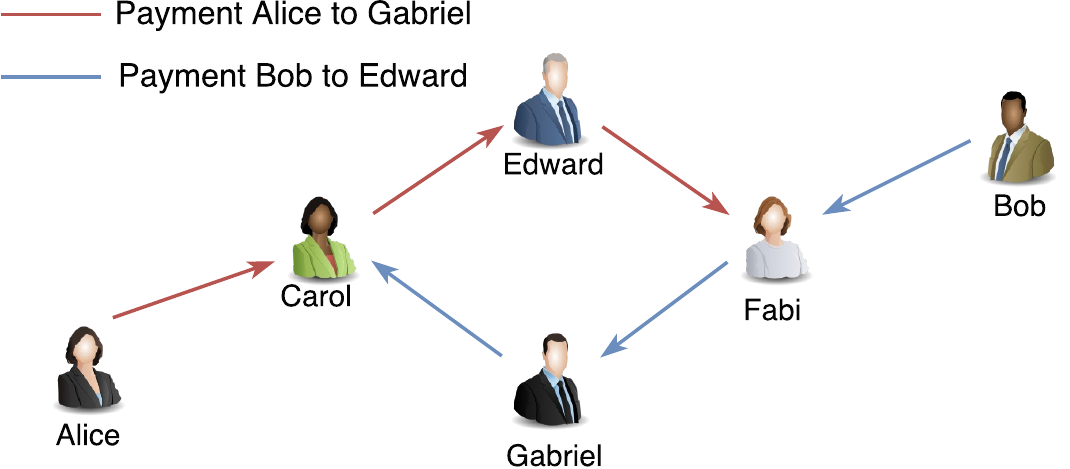}
\caption{ {\bf{Illustrative example of two blocking \payments:}} Alice to Gabriel (red) and 
Bob to Edward (blue). For simplicity, assume each \payment pays $1$ bitcoin and each 
payment channel has \capacity $1$\NEWPfor{SGAM}{170823}{ bitcoin}. Each 
\REPLACEPfor{SGAM}{170823}{link}{\pc} is colored with the \payment that has reached it 
first. In this deadlock situation, none of the \payments can continue further in the path and 
cannot be trivially completed. \label{fig:blocking-tx}}
\end{figure}

\pparagraph{Blocking \Payments (\sysname)} A best-effort solution for avoiding this \emph{deadlock} consists on letting both \payments fail. Aborted \payments do not affect 
the balance of the involved \users as the receiver would not disclose the release condition 
for the locked \pcs. Therefore, involved \pcs would get unlocked only after the corresponding timeout and 
bitcoins are sent back to 
the original owner.

The  sender of an aborted \payment can then 
randomly choose a waiting period to reissue the \payment. 
Although the \emph{blocking} mechanism closely resembles the practice of \users in  others 
payment networks such as Ripple~\cite{ripple-concurrency} or SilentWhispers~\cite{silentwhispers}, 
it might degrade transaction throughput in a 
fully decentralized  \pcn.

\pparagraph{Non-blocking Payments (\sysnamenonblocking)} An alternative 
\NEWPfor{SGAM}{170823}{solution} consists on a  \emph{non-blocking} solution where 
at least one out of a set of concurrent \payments completes. Our approach 
to achieve it assumes that there exists a global ordering of \payments (e.g., by 
a global \payment identifier). In a nutshell, \users can queue \payments with higher identifier than 
the current one ``in-flight'', and abort \payments with lower identifiers. This ensures that 
either the current in-flight \payment completes or one of the queued \payments would do, 
as their identifiers are higher. 
\REMOVEPfor{SGAM}{170518}{\TODOP{We could say here how this solves the deadlock 
in figure 4. It might repetitive and a waste of space. Opinions?}\MESSAGEMfor{SPGA}{170518}{let's briefly mention that}}

\REMOVEPfor{SGAM}{170519}{However, the requirement of a common payment identifier among all the \users in a 
payment path is at odds with anonymity, similar to what we described with the simple HTLC contract. 
We show that any possible non-blocking solution in a decentralized \pcn 
mandates such a global identifier. Therefore, there is 
an interesting,  inevitable tradeoff between the handling of concurrent \payments and the 
desired privacy guarantees. In this state of affairs, we define two alternative protocols:  
\sysname (blocking \payments) and \sysnamenonblocking (non-blocking \payments). }


%

\section{\sysname: Our Construction}
\label{sec:construction}

In this section, we  introduce the cryptographic building blocks required for our construction (\cref{sec:building-blocks}),  
we describe the details for the \newhtlc contract (\cref{sec:multihop-htlc}), 
we detail the constructions for \pcn operations (\cref{sec:protocol-details}), analyze its security and 
privacy (\cref{sec:fulgor-sec-priv-analysis}) and conclude  with a few 
remarks (\cref{sec:discussion}).

\pparagraph{Notation}
\REMOVEPfor{SGAM}{170519}{\TODOP{This notation paragraph seems from an old version to me}
We denote by $\lambda$ the security parameter of our system and we use
the standard definition for a \emph{negligible} function. We write $s
\gets S$ to denote the sampling of a random element $s$ from a set $S$. We
denote by $x||y$ the concatenation of two strings $x$ and $y$ and by
$[n]$ the set of integers $\{ 1, \dots, n\}$. We write the set of all
possible permutations of a string of length $n$ as $\mathcal{P}^n$.
We denote by $\ggroup$  a finite multiplicative cyclic group
and $g$ is a generator of $\ggroup$.}
\REPLACEPfor{SGAM}{170519}{Additionally, w}{
We denote by $\lambda$ the security parameter of our system and we use
the standard definition for a \emph{negligible} function. \REMOVEG{We denote by $\Delta$ the difference between 
two timeouts (\timeout, \timeout') associated to two HTLC contracts that ensures they are fulfilled sequentially. }
W}e denote by \ops the possible events in a \pc due to a \payment. 
The \ops  $\hold$ signals to lock the balance in the \pc corresponding to the \payment value. 
The \ops $\release$ signals the release of \REPLACEPfor{SGAM}{170823}{\capacity}{locked funds} in the \pc due to the abortion of a \payment.  
Correspondingly, the \ops $\confirm$ signals the  confirmation of a \payment accepted by 
the receiver. 

For ease of notation, we assume that \users identifiers $(\uid_i, \uid_{i+1})$ can be extracted 
from the channel identifier $\pair{\uid_i} {\uid_{i+1}}$.

\pparagraph{\NEWG{System }Assumptions}
We assume that every \user in the \pcn is aware of the complete network topology, that is, the set of all \users and the existence of a 
payment channel between every pair of \users. \NEWPfor{SGAM}{170519}{We further assume that the sender of a \payment chooses a \payment path to the 
receiver according to her own criteria.\REMOVEGfor{SAM}{170823}{ We consider path selection criteria orthogonal to this work.}} 
The current value on each \pc is not published but instead kept locally by the \users 
sharing a payment channel as otherwise privacy is trivially broken. We further assume that every \user is aware of the \payment fees charged 
by each other \user in the \pcn.

This can be accomplished in practice. The opening of a \pc between two \users requires to add a transaction in the blockchain that includes both \user identifiers. Therefore, the topology of the \pcn is trivially leaked. Moreover, the transaction used to open a \pc can 
contain user-defined data~\cite{bitcoin-opreturn}  
so that each \user can embed 
her own \payment fee. In this manner, each \user can proactively gather updated information about the network topology and fees from the blockchain itself 
or be disseminated by a gossip protocol~\cite{flare, ln-scalability-challenges}.


We  further assume
that pairs of users sharing a \pc communicate through secure and
authenticated channels (such as TLS), which is easy to implement given
that every user is uniquely identified by a public key. Also we assume
that the sender and the receiver of a (possibly indirect) transaction
can communicate through a secure and direct channel. Finally, 
we assume that the sender of a \payment can create an anonymous 
payment channel with each intermediate \user. The IP address where to 
reach  each \user could be encoded in the channel creation transaction and 
therefore logged in the blockchain. 
We note that our protocol  is completely
parametric with respect to the routing, therefore any onion
routing-like techniques would work in this context.

We consider the \emph{bounded synchronous} communication 
setting~\cite{AW04}.  In such communication model,  
time is divided into fixed communication rounds 
and it is assumed that all messages sent by a \user in a round are available to the 
intended recipient within a bounded number of steps in an execution. Consequently, absence of a message
indicates absence of communication from a \user during the round. 
In practice, 
this can be achieved with loosely synchronized clocks among the \users in the 
\pcn~\cite{ACS86}.


Finally, we assume that there is a total order \NEWPfor{SGAM}{170823}{among} the \users 
(e.g., lexicographically sorted by their public verification keys). 

\subsection{Building Blocks}
\label{sec:building-blocks}

\pparagraph{Non-Interactive Zero-Knowledge} Let
$\relation: \{ 0,1\}^* \times \{0,1\}^* \to \{0,1\}$ be an NP
relation, and let $\lang$ be the set of positive instances for
$\relation$, i.e.,
$\lang = \{x \mid \exists w \text{ s.t. } \relation(x,w) =1\}$. A
non-interactive zero-knowledge proof for $\relation$ consists of a
single message from a prover $\prover$ to a verifier $\verifier$. The
prover $\prover$ wants to \NEWGfor{SAM}{170823}{compute a proof $\pi$ that} convinces the verifier $\verifier$ that a
certain statement $x \in \lang$. We allow the prover to run on an
extra private input $w$ such that $\relation(x,w) = 1$. The verifier
can either accept or reject, depending on $\pi$. A
$\nizk$ is complete if the $\verifier$ always accepts honestly
computed $\pi$ for a statement $x \in \lang$ and it is sound if
$\verifier$ always rejects any $\pi$ for all $x \not\in \lang$,
except with negligible probability. Loosely speaking, a $\nizk$ proof
is zero knowledge if the verifier learns nothing from $\pi$ beyond the
fact that $x \in \lang$. Efficient $\nizk$ protocols are known to
exist \REMOVEGfor{SAM}{170616}{in the common reference string model~\cite{grothsahai} and }in the
random oracle model~\cite{zkboo}.

\pparagraph{Two \Users Agreement} Two \users $\uid_i$ and $\uid_j$ sharing a \pc,    
locally maintain the state of the \pc defined as a scalar $\pcstate := \capa(\pair{\uid_i}{\uid_{j}})$ that  
denotes the current \capacity of their \pc. We require a two party agreement 
protocol that ensures that both \users agree on the current value of $\capa(\pair{\uid_i} {\uid_{j}})$ 
at each point in time. We describe the details of such protocol in~\cref{sec:agreement-two-users}. 
For readability, in the rest we implicitly assume that two \users sharing a \pc satisfactorily agree 
on its current state.  


\subsection{Multi-Hop HTLC}
\label{sec:multihop-htlc}

\REMOVEPfor{SGAM}{170519}{In this section we describe our proposal for payments based on Hash
Time-Lock Contracts (HTLCs). This routine is at the core of our
construction for a \pcn and we describe it as a separate building block
since we believe that such a technique might be of independent
interest. \MESSAGEPfor{PSGAM}{170519}{We have said this couple of times already. 
We can remove previous paragraph if space required}}

We consider the standard scenario of an indirect payment
from a sender $\sender$ to a receiver $\receiver$ for a certain value
$\val$ through a path of users $(\uid_1, \dots, \uid_n)$, where
$\uid_n =\receiver$. All users belonging to the same network share the
description of a hash function
$H: \{0,1\}^* \to \{0,1\}^\lambda$ that we model as a random oracle. 

Let $\lang$ be the following language: $\lang = \{(H, y', y, x) \mid \exists (w) \text{ s.t. } y' \allowbreak = H(w) \land
y  = H(w \oplus x)\}$
where $w \oplus x$ denotes the bitwise XOR of the two
bitstrings. Before the payment starts, the sender $\sender$ locally
executes the following $\Setup_\text{HTLC}$ algorithm described in~\cref{fig:setup-htlc}.

\begin{figure}[b]
{\fbox{
    \procedure[codesize=\small]{$\Setup_\text{HTLC}(n):$}{
      \forall i \in [n]:\\
      \quad x_i \in \{0,1\}^{\lambda};   y_i \gets H\left(\bigoplus^{n}_{j=i}x_j\right)\\
      \forall i \in [n-1]:\\
      \quad \pi_i \gets \prover\left((H, y_{i+1}, y_i, x_i),
        \left(\bigoplus^{n}_{j=i+1}x_j\right)\right)\\
      \pcreturn ((x_1, y_1, \pi_1), \dots, (x_n,y_n))
    } } }
\caption{Setup operation for the \newhtlc contract. 
\label{fig:setup-htlc}}
\end{figure}
Intuitively, the sender samples
$n$-many random strings $x_i$ and defines $y_i$ as
$H\left(\bigoplus^{n}_{j=i}x_j\right)$ which is the XOR combination of
all $x_j$ such that $j \ge i$. Then, $\sender$ computes the proofs
$\pi$ to guarantee that each $y_i$ is well-formed, without revealing
all of the $x_i$. The receiver is provided with $(x_n, y_n)$ and she 
simply checks that $y_n = H(x_n)$. $\sender$ then sends
$(x_i, y_i, \pi_i)$ to each intermediate user $\uid_i$, through a
direct communication channel. Each $\uid_i$ runs
$\verifier((H, y_{i+1}, y_i, x_i),\pi_i)$ and aborts the payment if
the verification algorithm rejects the proof.

Starting from the user $\uid_0 = \sender$, each pair of users
$( \uid_i, \uid_{i+1})$ check whether both users received the same
values of $(y_{i+1}, \val)$. This can be done by simply exchanging and
comparing the two values. If this is the case, they establish 
HTLC ($\uid_i$, $\uid_{i+1}$, $y_{i+1}$, $\val$, $t_i$) 
 as described in~\cref{sec:current-pcns-background},  
where $t_i$ defines some timespan such that for all
$i \in [n]: t_{i-1} = t_i + \Delta$,\NEWG{ for some positive value $\Delta$}.
Once the contract between
$(\uid_{n-1}, \uid_n)$ is settled, the user $\uid_n$ (the receiver)
can then pull $\val$ bitcoins by releasing the $x_n$, which by
definition satisfies the constraint $H(x_n) = y_n$. Once the value of
$x_n$ is published, $\uid_{n-1}$ can also release a valid condition
for the contract between $(\uid_{n-2}, \uid_{n-1})$ by simply
outputting $x_{n-1} \oplus x_n$. In fact, this mechanism propagates
for every intermediate user of the payment path, until $\sender$: For
each node $\uid_i$ it holds that, whenever the condition for the
contract between $(\uid_i,\uid_{i+1})$ is released, i.e., somebody
publishes a string $r$ such that $H(r) = y_{i+1}$, then $\uid_i$
immediately learns $x_i \oplus r$ such that $H(x_i \oplus r) = y_i$,
which is a valid condition for the contract between
$(\uid_{i-1},\uid_{i})$. It follows that each \NEWPfor{SGAM}{170823}{intermediate} user whose
\emph{outgoing} contract has been pulled is able to release a valid
condition for the \emph{incoming} contract\REMOVEPfor{SGAM}{170823}{, except for $\sender$}.

\REMOVEGfor{SAM}{170823}{We present a variation of this protocol based on discrete logarithm-hard groups in Appendix~\ref{sec:dltc} and we discuss the privacy implications of confidential transactions in \pcns in Appendix~\ref{sec:ct}.}



\subsection{Construction Details}
\label{sec:protocol-details}

In the following, we describe the details of the three
operations (\openc, \closec,  \pay) that compose \sysname.

\begin{asparaitem}
\item $\openc(\uid_1, \uid_2, \beta, \timeout, \fee)$: 
The purpose of this operation is to open a payment channel between 
\users $\uid_1$ and $\uid_2$. 
For that, they create an initial Bitcoin deposit  \REMOVEGfor{SAM}{170512}{that includes the verification keys $\vk_1, \vk_2$ } \MESSAGEGfor{SAM}{170512}{We never mention digital signatures, either we introduce them explicitly or we don't mention them at all}
that includes the following information: their Bitcoin \accs, the initial \capacity of the channel ($\beta$), the channel  timeout (\timeout), 
the fee charged to use the channel (\fee) and a channel identifier 
($\pair{\uid_1}{\uid_2}$) agreed beforehand between both \users.  
After the Bitcoin deposit has been 
successfully added to the blockchain, the operation returns $1$. If any of the previous steps is not carried out 
as defined, the operation returns $0$.

\item $\closec(\pair{\uid_1}{\uid_2}, \REPLACEG{\cstate}{\val})$:  
This operation is used by two \users ($\uid_1, \uid_2$) sharing an open payment channel ($\pair{\uid_1}{\uid_2}$)
to close it at the state defined by \REPLACEG{\cstate}{$\val$} and accordingly update 
their bitcoin balances in the Bitcoin blockchain. 
This operation in \sysname 
is performed as defined in the original proposal of payment channels (see~\cref{sec:pay-channels}), 
additionally returning $1$ if and only if the corresponding Bitcoin transaction is added to the Bitcoin blockchain.

\item $\pay((\pair{\uid_0}{\uid_1}, \ldots, \pair{\uid_n}{\uid_{n+1}}), \val)$: 
  A payment operation transfers a value $\val$ from a
  sender  ($\uid_0$) to a receiver  ($\uid_{n+1}$) through a path of open payment channels between
  them (\pair{\uid_0}{\uid_1}, \ldots, \pair{\uid_n}{\uid_{n+1}}). Here, we describe a \emph{blocking} version of the \payment operation (see~\cref{sec:approach}). 
  We discuss the non-blocking version of the \payment operation in~\cref{sec:concur-payments}.


  \iftechreport
  \begin{figure}
  \else
   \begin{figure}[tb]
  \fi
  \setlength\columnsep{1pt}

     {
     \begin{mdframed}
      \procedure[codesize=\small]{$\pay_{\uid_0}(m):$}{
         (\hl{\Txid},  \set{\pair{\uid_0}{\uid_1}} \cup \set{\pair{\uid_i} {\uid_{i+1}}}_{i\in[n]}, \val) \gets m\\
      	  \val_{1} := \val + \sum_i^n \text{fee}(\uid_i)\\
	  \pcif  \val_{1} \leq \capa(\pair{\uid_0} {\uid_1}) \pcthen \\
	  \quad \capa(\pair{\uid_0} {\uid_1}) :=  \capa(\pair{\uid_0} {\uid_1}) - \val_{1}\\
	  \quad \timeout_0 := \timeout_{\text{now}} + \Delta \cdot n\\
       	  \quad \forall i \in [n]:\\
           \quad \quad \val_i := \val_{1} - \sum_{j=1}^{i-1} \text{fee}(\uid_j)\\
           \quad \quad \timeout_i := \timeout_{i-1} - \Delta\\
          \quad \quad \set{(x_i, y_i, \pi_i)}_{i\in[n+1]} \gets \Setup_\text{HTLC}(n+1)\\
        \quad \quad \send(\uid_i, ((\hl{\Txid}, x_i, y_i, y_{i+1}, \\ 
                      \quad \quad \quad \pi_i, \pair{\uid_{i-1}} {\uid_i}, \pair{\uid_i} {\uid_{i+1}}, \val_{i+1}, \timeout_i, \timeout_{i+1}), \hold) )\\
        \quad \text{{\bf HTLC}}(\uid_0, \uid_1, y_1, \val_{1}, t_{1})\\
        \quad \send(\uid_{n+1}, (\hl{\Txid,} x_{n+1}, y_{n+1}, \pair{\uid_n} {\uid_{n+1}}, \\
        \quad \quad \quad \val_{n+1}, \timeout_{n+1}))\\
        \pcelse\\
        \quad \textbf{abort}
        }
        \end{mdframed}
    }
 \caption{ The \pay routine in \sysname for the sender. The \colorchange pseudocode shows additional steps required in \sysnamenonblocking. \label{fig:pay-operations} }    
\end{figure}

 \iftechreport
  \begin{figure}
  \else
   \begin{figure}[tb]
  \fi
	  {
	\begin{mdframed}
        \procedure[codesize=\small]{$\pay_{\uid_{n+1}}(m):$}{
        (\hl{\Txid,} x_{n+1}, y_{n+1}, \pair{n}{n+1}, \val_{n+1}, \timeout_{n+1}) \gets m\\
        \pcif H(x_{n+1}) = y_{n+1} \textbf{ and } \timeout_{n+1} > \timeout_{\text{now}} + \Delta \pcthen\\
         \quad \textbf{store } (x_{n+1}, y_{n+1}, \pair{n}{n+1}, t_{n+1} )\\
         \quad \send (\uid_n, ((\hl{\Txid}, x_{n+1}, y_{n+1}, \pair{n}{n+1}), \confirm))\\
        \pcelse\\
        \quad \send (\uid_n, ((\hl{\Txid,} y_{n+1}, \pair{n}{n+1}, \val_{n+1}), \release)) 
      }
      \end{mdframed}  
      }

     {
     \begin{mdframed}
      \procedure[codesize=\small]{$\pay_{\uid_i}(m):$}{
      (m^*, \ops) \gets m\\
      \pcif \ops = \hold \pcthen\\
      \quad (\hl{\Txid}, x_i, y_i, y_{i+1}, \pi_i, \pair{i-1} {i}, \pair{i}{i+1}, \\ 
      \quad \quad \quad \val_{i+1}, \timeout_i, \timeout_{i+1}) \gets m^* \\
      \quad \pcif  \val_{i+1}  \leq \capa(\pair{\uid_i} {\uid_{i+1}}) \textbf{ and } \verifier((H, y_{i+1}, y_i, x_i),\pi_i) \\ 
      \quad \quad \textbf{and } \timeout_{i+1} =  \timeout_i - \Delta \pcthen \\
         \quad \quad \capa(\pair{\uid_{i}} {\uid_{i+1}})  := \capa(\pair{\uid_i} {\uid_{i+1}})  - \val_{i+1}\\
         \quad \quad \text{{\bf HTLC}}(\uid_i, \uid_{i+1}, y_{i+1}, \val_{i+1}, \timeout_{i+1})\\
         \quad \quad \hl{\cur(\pair{\uid_i} {\uid_{i+1}}).\add(m^*)}\\
        \quad \hl{\pcelse \pcif \exists k~|~\Txid > \cur(\pair{\uid_i} {\uid_{i+1}})[k].\Txid \pcthen }\\ 
        \quad \quad \hl{\q(\pair{\uid_i} {\uid_{i+1}}).\add(m^*)}\\
        \quad \pcelse\\
        \quad \quad \send (\uid_{i-1}, ((\hl{\Txid,} y_{i}, \pair{i-1} {i}, \val_i), \release)) \\
        \pcelse \pcif \ops = \release \pcthen\\
        \quad  (\hl{\Txid,} y_{i+1}, \pair{i} {i+1}, \val_{i+1}) \gets m^*\\
        \quad \capa(\pair{\uid_i}{\uid_{i+1}}) :=  \capa(\pair{\uid_i} {\uid_{i+1}}) + \val_{i+1}\\
        \quad \send(\uid_{i-1}, ((\hl{\Txid,} y_{i}, \pair{i-1} {i}, \val_{i}), \release))\\
        \quad \hl{\cur(\pair{\uid_i} {\uid_{i+1}}).\delete(m^*.\Txid) }\\
        \quad \hl{m' := \max(\q(\pair{\uid_i}{\uid_{i+1}}))} \\
        \quad \hl{\pay_{\uid_i}((m' ,\hold))}\\
        \pcelse \pcif \ops = \confirm \pcthen\\
        \quad (\hl{\Txid,} x_{i+1}, y_{i+1}, \pair{i}{i+1}, \val_{i+1}) \gets m^*\\\
        \quad \textbf{store } (x_{i+1} \oplus x_i, y_i, \pair{i-1}{i}, \timeout_i)\\
         \quad \send(\uid_{i-1}, ((\hl{\Txid,} x_{i+1} \oplus x_i, y_{i}, \pair{i-1}{i}, \val_{i}), \confirm))\\
        \quad \hl{\cur(\pair{\uid_i} {\uid_{i+1}}).\delete(m^*.\Txid) }
      }
    \end{mdframed}
    }
    
  \caption{The \pay routine in \sysname for the receiver 
  and each intermediate \user. The \colorchange pseudocode shows additional steps in \sysnamenonblocking. $\max(\q)$ returns the information for the \payment with highest identifier 
  among those in \q. \label{fig:pay-operations-two} }    
  \end{figure}

  \REPLACEPfor{SGAM}{170823}{T}{As shown in~\cref{fig:pay-operations} (black pseudocode), t}he sender first calculates the cost of sending \val bitcoins 
  to \receiver as $\val_1 := \val + \sum_i \textit{fee}(\uid_i)$, and the corresponding cost at each of the 
  intermediate hops in the \payment path. If the sender does not have enough bitcoins, she aborts the 
  payment. Otherwise,  the sender sets up the contract for each intermediate \pc following the mechanism 
  described in~\cref{sec:multihop-htlc} and sends the information to the corresponding \users.

  Every intermediate \user verifies
  that the incoming HTLC has an associated value smaller or equal than
  the \capacity of the \pc with her sucessor in the path. Additionally, every intermediate \user 
  verifies that the zero-knowledge proof associated to the HTLC for incoming and outgoing \pcs 
  correctly verifies and that the timeout for the incoming HTLC is bigger than the timeout for the outgoing HTLC by a difference of $\Delta$. If so, she generates the
  corresponding HTLC for the same associate value (possibly minus the fees) 
  with the successor \user in the path; otherwise, she
  aborts by triggering the \release event to the predecessor \user in the path. 
 These operations have been shown in~\cref{fig:pay-operations-two} (black pseudocode).
 
 If every \user in the path accepts the \payment, it
  eventually reaches the receiver who in turn releases the information 
  required to fulfill the HTLC contracts in the path\NEWPfor{SGAM}{170823}{ (see~\cref{fig:pay-operations-two} (black pseudocode))}. Interestingly, if any intermediate \user
  aborts the \payment, the receiver does not release the condition as she does not
  receive any \payment\REPLACEPfor{SGAM}{170823}{ and }{. Moreover, }\pcs already set in the previous hops
  of the path are voided after the  timeout set in the corresponding HTLC. 


\end{asparaitem}

\subsection{Security and Privacy Analysis}
\label{sec:fulgor-sec-priv-analysis}
In the following, we state the security and privacy results for \sysname. We prove our results in the 
$(\ttp_\mathsf{anon}, \ttp_\blockchain)$-hybrid model. In other words, Theorem~\ref{thm:uc-security} holds for any UC-secure realization 
of $\ttp_\mathsf{anon}$ and $\ttp_\blockchain$. We show the proof of Theorem~\ref{thm:uc-security} in \cref{sec:uc-security-proof}.

\begin{theorem}[UC-Security]
\label{thm:uc-security}
Let $H: \bit{*} \to \bit{\lambda}$ be a hash function modelled as a random oracle, and let $(\prover, \verifier)$ a zero-knowledge proof system, then \sysname UC-realizes the ideal functionality $\ttp$ defined in~\cref{fig:blocking-tx} in the $(\ttp_\mathsf{anon}, \ttp_\blockchain)$-hybrid model.
\end{theorem}

\subsection{System Discussion}
\label{sec:discussion}

\pparagraph{Compatibility with Bitcoin} We note that all of the non-standard
cryptographic operations (such as $\nizk$ proofs) happen off-chain,
while the only algorithm required to be executed in the verification
of the blockchain is the hash function $H$, which can be instantiated
with SHA-256. Therefore our Multi-Hop HTLC scheme and \sysname as a whole 
is fully compatible
with the current Bitcoin script. \REPLACEPfor{SGAM}{170616}{ 
\TODOA{Mention that current Bitcoin need SegWit or similar solution.}}
{Moreover, as mentioned in~\cref{sec:pay-channels}, the addition of 
SegWit or similar solution for the malleability issue in Bitcoin fully enables  \pcs 
in the Bitcoin system~\cite{segwit-pcn-bitcoin-magazine}. }

\pparagraph{Generality} 
\sysname is general to \pcns (and not only tied to Bitcoin). 
\sysname requires that: (i) \openc allows to embed custom data (e.g., fee); 
(ii) conditional updates of the balance in the \pc. As arbitrary data 
can be included in cryptocurrency transactions~\cite{bitcoin-opreturn} and 
most \pcns support, among others, the HTLC contract, \sysname can 
be used in many other \pcns such as Raiden, a \pcn for Ethereum~\cite{raiden-nw}.
\REPLACEPfor{SGAM}{170616}{
\TODOA{Mention that the DLog contract explained appendix can be better for Ethereum etc}}
{
\REMOVEGfor{SAM}{170823}{Interestingly,  in \pcns with a richer scripting language (e.g., Raiden), more efficient 
contracts are possible such as the Discrete-Log Time Lock Contract (DTLC) that 
we describe in ~\cref{sec:dltc}.}
}

\REMOVEPfor{SGAM}{170519}{\pparagraph{Decentralization} 
\sysname operations require only the intervention of the \users owning the \pcs involved in 
each of the operations. 
In particular, \openc and \closec are performed by the two \users sharing the \pc, while 
the \pay operation is jointly carried out by the \users in the \payment path 
from the sender to the receiver. Moreover, the current state of each \pc is maintained 
by the two \users connected by it. Therefore, \sysname does not require any 
 third party provider and is fully decentralized. 
}

\NEWPfor{SGAM}{170519}{
\pparagraph{Support for Bidirectional Channels} \sysname can be easily extended 
to support bidirectional \pcs and only two minor changes are required. First, the 
\payment information must include the direction requested at each \pc. 
Second, the capacity of a channel $\pair{\uid_L}{\uid_R}$ is a tuple of 
values $(L, R, T)$ where $L$ denotes the current balance for $\uid_L$, $R$ is 
the current balance of $\uid_R$ and $T$ is the total capacity of the channel.  A 
\payment from left to right for value $\val$ is possible if $L \ge \val$ and $R + \val \le T$. 
In such case, the tuple is updated to $(L-\val, R+\val, T)$. A \payment from 
right to left is handled correspondingly. 

}

\REMOVEPfor{SGAM}{170519}{\pparagraph{Blocking \Payments}  Concurrent \payments over disjoint paths 
in the \pcn are trivially handled by \sysname. However, the current definition of 
the \pay operation in \sysname does not avoid that concurrent \payments requiring 
more \capacity than available at \emph{shared links} end up in a deadlock 
situation where none of the \payments goes through (see \cref{fig:blocking-tx}). 
In~\cref{sec:concur-payments}, we describe how to extend \sysname to ensure 
that at least one \payment goes through, thereby avoiding deadlocks.
\MESSAGEPfor{PSGAM}{170519}{We explained exactly this in the overview. We can omit paragraph above  
if space needed}}

\section{Non-blocking Payments  in \pcns}
\label{sec:concur-payments}

In this section, we discuss how to handle concurrent \payments in a non-blocking manner. In other words, 
how to guarantee that at least one \payment out of a set of concurrent \payments terminates.  

In the following, we start with an impossibility result that dictates the design  of \sysnamenonblocking, our protocol for 
non-blocking \payments. Then, we describe the modifications required in the ideal world functionality and 
\sysname to achieve them. Finally, we discuss the implications of these modifications in terms of privacy properties. 

\subsection{Concurrency vs Privacy}


We show that achieving non-blocking progress requires  a global state 
associated to each of the \payments. Specifically, we show that we cannot provide
\emph{disjoint-access parallelism} and non-blocking progress for \pcns.
Formally, a \pcn implementation is \emph{disjoint-access parallel} if for any two payments channels $e_i,e_j$,
\pcstate(e$_i$) $\cap$ \pcstate(e$_j$) =$\emptyset$.
\begin{lemma}
\label{lemma:concurtid}
\label{lemma:transaction-id}
There does not exist any strictly serializable disjoint-access parallel implementation for the payment channels problem that provides non-blocking progress.
\end{lemma}
We defer to \cref{sec:concurrency-proofs-lemmas} for \REPLACEAfor{SGM}{170613}{the proof}{a proof sketch}.    
Having established this inherent cost to concurrency and privacy, we model global state 
 by a \Txid field attached to each of the \payments. \REPLACEPfor{SGAM}{170823}{This}{We remark that 
 this \Txid}, however, 
allows an adversary to reduce the set of possible senders and receivers for the \payment, therefore inevitably reducing the 
privacy guarantees, as we discuss in~\cref{sec:ideal-world-nonblocking}.
\subsection{Ideal World Functionality}
\label{sec:ideal-world-nonblocking}
Here, we show how to modify the ideal functionality $\ttp$, as described in \cref{sec:ideal-world}, to account for the changes to achieve non-blocking progress in any \pcn. First,  a single identifier $\Txid$ (as opposed to independently sampled $h_i$) is used for all the \pcs in the path 
$(\pair{\uid_0}{\uid_1}, \dots, \pair{\uid_n}{\uid_{n+1}})$. Second, $\ttp$ no longer aborts a \payment simply when no capacity is left in a \pc. 
Instead, $\ttp$ queues the \payment if its \Txid is higher than the current in-flight \payment, or aborts it the $\Txid$ is lower. We detail the modified 
ideal functionality in~\cref{sec:ideal-world-non-blocking}.

\pparagraph{Discussion}
\TODOS{We need the serializability in the ideal world and an argument why the ideal world 
achieves it}
Here, we discuss how the modified ideal world definition captures the security and privacy notions of interest 
as described in~\cref{sec:attacker-model}. In particular, it is easy to see that 
the notions of \emph{\correctbal} and \emph{value privacy} are enforced along the same lines. 
However, the leakage of the same 
\payment identifier among all intermediate \users in the payment path, reduces the possible set of sender and receivers to 
the actual sender and receiver  
for such \payment, thereby breaking relationship anonymity.
Therefore, there is an inherent tradeoff between how to handle concurrent \payments (blocking or non-blocking) and 
the anonymity guarantees.

\begin{figure}[tb]
\includegraphics[width=\columnwidth]{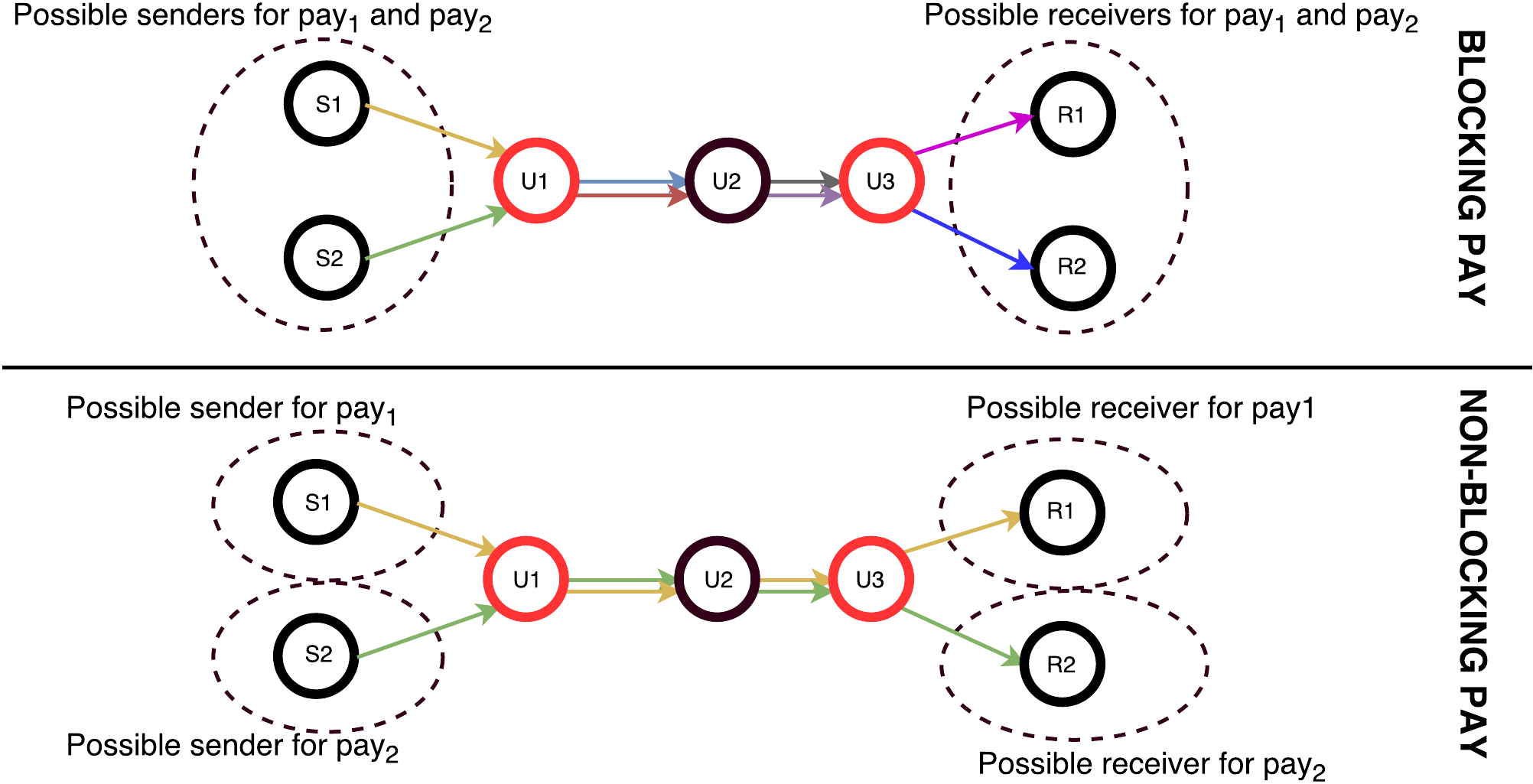}
\caption{Illustrative example of tradeoff between concurrency and privacy. Each node represents a user: black nodes are honest and 
red are byzantine. In both cases, we assume two concurrent \payments: $S_1$ pays $R_1$ and $S_2$ pays $R_2$ through the path $U_1, U_2, U_3$. 
The color of the arrow denotes the payment identifier. Dashed ellipses denote the anonymity set for each case.
\label{fig:tradeoff-privacy}}
\end{figure}

An illustrative example of this tradeoff is shown in~\cref{fig:tradeoff-privacy}. It shows how 
two simultaneous payments $\pay_1((\pair{S_1}{U_1}, \pair{U_1}{U_2},$ $\pair{U_2}{U_3}, \pair{U_3}{R_1}), \val)$ and $\pay_2((\pair{S_2}{U_1}, \pair{U_1}{U_2},$ $\pair{U_2,}{U_3}, \pair{U_3}{R_2}),$ $\val)$ 
are handled depending on whether concurrent payments are handled in a blocking or non-blocking fashion. 
We assume that both payments can successfully finish in the current \pcn and that both payments transfer 
the same payment amount \val, as otherwise relationship anonymity is trivially broken.  

For the case of blocking payments, each intermediate \user $\uid_j$ 
observes an independently chosen identifier $\Txid_{ij}$ for each \payment $\pay_i$. Therefore, 
the attacker is not able to correlate the pair $(\Txid_{11}, \Txid_{21})$ (i.e., view of $U_1$) with the pair 
$(\Txid_{13}, \Txid_{23})$ (i.e., view of $U_3$). It follows that for a \pay operation issued by any node, say $S_1$, the set of possible receivers that the adversary observes is $\set{R_1, R_2}$. 

However, when the concurrent payments are handled in a non-blocking manner, the adversary  
observes for $\pay_1$ that $\Txid_{11} = \Txid_{13}$. Therefore, the adversary can trivially derive that 
the only possible receiver for a \pay initiated by $S_1$ is $R_1$.

\TODOM{we should add a sentence saying who is supposed to choose between blocking and non-blocking(the sender?) and that this flexibility is a feature of the system}

\subsection{\sysnamenonblocking: Our Construction}

\REMOVEPfor{SGAM}{170519}{In this section, we describe how to modify \sysname to handle concurrent \payments in a non-blocking manner.}

\pparagraph{Building Blocks}
We require the same building blocks as described in \cref{sec:building-blocks} and \cref{sec:multihop-htlc}. 
The only difference is that the channel's state between two \users is now defined 
as  $\pcstate := (cur_{(\uid_i, \uid_j)}[\,],$ $\q_{(\uid_i, \uid_j)}[\,], \valtext_{(\uid_i, \uid_j)})$, where $\cur$ denotes an array of \payments 
 currently using (part of) the \capacity available  at the \pc;  
$\q$ denotes the array of \payments waiting for enough \capacity at the \pc, and \valtext denotes the current \capacity value of the \pc.

\pparagraph{Operations} 
The \openc and \closec operations remain as described in \cref{sec:protocol-details}. 
However, the \pay operation has to be augmented to ensure non-blocking \payments. We have described 
the additional actions in \colorchange pseudocode in \cref{fig:pay-operations,fig:pay-operations-two}. 

In the following, we informally describe these additional actions required for the \pay operation.
In a nutshell, when a \payment reaches an intermediate \user in the path, 
several events can be triggered. 
The simplest case is when the corresponding \pc is not saturated yet (i.e., enough \capacity is left for 
the \payment to succeed). The user accepts the \payment and 
simply stores its information in \cur as an in-flight \payment.  

The somewhat more interesting case occurs when the \pc is saturated. 
This means that (possibly several) \payments have been already gone through 
the \pc. In this case, the simplest solution is to abort the new \payment, but this leads to deadlock 
situations. Instead, we ensure that deadlocks do not occur by leveraging the 
total order of \payment identifiers: 
If the new \payment identifier ($\Txid$) is higher than any of the  \payment identifiers currently 
active in the \pc (i.e., included in \cur[\,]), \NEWPfor{SGAM}{170823}{the payment identified by} \Txid is stored in $\q$. In this manner, if any of the currently 
active \payments are aborted, 
a queued \payment ($\Txid^*$) can 
be recovered from $\q$ and reissued towards the receiver.  On the other hand, if  
 \Txid is lower than every identifier for currently active \payments, 
 \NEWPfor{SGAM}{170823}{the payment identified by} \Txid is directly 
aborted as it would \REPLACEPfor{SGAM}{170823}{never}{not} get to complete in the presence of a concurrent 
\payment with higher identifier in the \pcn.


\subsection{Analysis and System Discussion}

\pparagraph{Security and Privacy Analysis}
In the following, we state the security and privacy results for \sysnamenonblocking when handling 
\payments in a non-blocking manner. We prove our results in the 
$(\ttp_\mathsf{anon}, \ttp_\blockchain)$-hybrid model. In other words, Theorem~\ref{thm:uc-security-non-blocking} holds for any UC-secure realization 
of $\ttp_\mathsf{anon}$ and $\ttp_\blockchain$ (\REPLACEAfor{SGM}{170613}{proof}{analysis} in  \cref{sec:uc-security-proof}).

\begin{theorem}[UC-Security]
\label{thm:uc-security-non-blocking}
Let $H: \bit{*} \to \bit{\lambda}$ be a hash function modelled as a random oracle, and let $(\prover, \verifier)$ a zero-knowledge proof system, then 
\sysnamenonblocking UC-realizes the ideal functionality $\ttp$ described in~\cref{fig:ideal-world-non-blocking} in the $(\ttp_\mathsf{anon}, \ttp_\blockchain)$-hybrid model.
\end{theorem}

\pparagraph{System Discussion}
\NEWPfor{SGAM}{170519}{
\sysnamenonblocking is compatible with Bitcoin, can be generally applicable to \pcn and supports 
bidirectional \pcs similar to \sysname. 
}
Moreover, the \sysnamenonblocking protocol  provides non-blocking progress.
Specifically, \sysnamenonblocking ensures that some \payment  successfully terminates in every execution.
Intuitively, this is because any two conflicting \payments can necessarily be ordered by their respective unique identifier:
the highest \payment identifier is deterministically identified and terminates successfully while the lower priority \payment \emph{aborts}.
%
%
\renewcommand{\arraystretch}{1.2}
\begin{table}[b]
\caption{Comparison between \sysname and \sysnamenonblocking. \label{table:tradeoff-systems}}
\small
\begin{tabular}{c | c | c}
 & {\bf \sysname }& {\bf \sysnamenonblocking }\\
 \hline
 {\bf \Correctbal} & \CIRCLE & \CIRCLE\\
 {\bf Serializability} & \CIRCLE & \CIRCLE\\
 {\bf Non-blocking progress} & \Circle & \CIRCLE\\
 {\bf Value Privacy } & \CIRCLE & \CIRCLE\\
 {\bf Anonymity} & \CIRCLE & \LEFTcircle
\end{tabular}

\end{table}

\subsection{\sysname vs \sysnamenonblocking}
In this work, we characterize the tradeoff between the two protocols presented in this work. 
As shown in~\cref{table:tradeoff-systems}, both protocols guarantee crucial security 
and correctness properties such as \correctbal 
and serializability. By design, \sysnamenonblocking is the only protocol that 
ensures non-blocking progress. Finally, regarding privacy, we aimed at 
achieving the strongest privacy possible. However, although both protocols 
guarantee value privacy, we have shown that it is impossible to 
simultaneously achieve non-blocking progress and  strong anonymity. Therefore, 
\sysname achieves strong anonymity while \sysnamenonblocking achieves non-blocking progress \NEWPfor{SGAM}{170616}{at 
the cost of weakening the anonymity guarantees. We note nevertheless that \sysnamenonblocking provides relationship anonymity 
only if none of the intermediate nodes is compromised. Intuitively, \sysnamenonblocking provides 
this (weaker) privacy guarantee because it still uses \newhtlc as \sysname. 
}



%

\section{Performance Analysis}
\label{sec:implementation}

In this section, we first evaluate the performance of \sysname. 
Finally, 
we describe the overhead required for \sysnamenonblocking.

We have developed a proof-of-concept implementation in Python to evaluate the 
performance of \sysname. 
We interact with the API of \emph{lnd}~\cite{ln-release}, the recently released Lightning Network implementation, 
We  use 
\emph{listchannels} to extract the current \capacity of an open payment channel,  
\emph{listpeers} to extract the list of public keys from other \users in the network, and 
\emph{getinfo} to extract the \user's own public key.  
We have instantiated the hash function with SHA-256. We have implemented the \newhtlc 
using a python-based implementation 
of ZK-Boo~\cite{pyzkboo} to create the zero-knowledge proofs. We set ZK-Boo to use SHA-256, 136 rounds 
to achieve a soundness error of the proofs of $2^{-80}$, and a witness of $32$ bytes  as in~\cite{htlc-snarks}.

\pparagraph{Implementation-level Optimizations}
During the protocol description, we have assumed that the sender creates a different 
anonymous communication channel with each intermediate \user. In our implementation, 
however, we use Sphinx~\cite{sphinx} 
to create a single anonymous communication channel between sender and receiver, 
where intermediate nodes are the intermediate \users in the path. Sphinx allows to  
send the required \payment information to each intermediate \user while obfuscating 
the information intended for other \users in the path and the actual length of the path by 
padding the forwarded data. This optimization has 
been discussed in the bitcoin community and implemented in the current release of 
\emph{lnd}~\cite{bolt4}. 


\pparagraph{Testbed}
We have simulated five \users and created a linear structure of \pcs: \user $i$ has \pcs 
open only with \user $i-1$ and \user $i+1$, \user $0$ is the sender, and \user $4$ is 
the receiver of the \pay operation. 
We run each of the \users  in a separated virtual machine with 
an  Intel Core i7 3.1 GHz processor and 2 GB RAM. The machines are connected 
in a local network with a mean latency of 111.5 milliseconds. 
For our experiments, we assume that 
each \user has already opened the corresponding \pcs and got the public verification key of each other \user 
in the \pcn. As this is a one time setup operation, we do not account for it in our experiments.

\pparagraph{Performance} 
%
%
 We have first executed the payment operation available in the 
\emph{lnd} software, 
which uses the HTLC-based \payment as the contract for conditional updates in a \pc. 
We observe that a (non-private) \pay operation over a path with 5 \users takes $609$ ms 
and so needs \sysname. 
Additionally, the \sender must  run the $\Setup_\text{HTLC}(n+1)$ protocol, increasing thereby her 
computation time. Moreover, the \sender must send the additional information corresponding to 
the \newhtlc contract (i.e., $(x_i, y_i, y_{i+1},  \pi_i)$) 
to each intermediate \user, which adds communication complexity. 

The sender requires $309$ ms to compute 
the proof $\pi_i$ for each of the 
intermediate \users. Each proof is of size $1.65$ MB. Finally, each  intermediate \user  
requires $130$ ms to verify  $\pi_i$. 
We focus on the zero-knowledge proofs as they are the most expensive operation.

Therefore, the total computation overhead is $1.32$ seconds (\emph{lnd} \pay and \newhtlc) and the total communication overhead is 
less than $5$ MB (3 zero-knowledge proofs plus the tuple of small-size values $(x_i, y_i, y_{i+1})$ per intermediate \user). 
We observe that previous proposal~\cite{htlc-snarks} required around $10$ seconds to compute only a single 
zero-knowledge proof. In contrast, the \pay operation in \sysname requires less than $2$ seconds of computation and to communicate less 
than $5$ MB among the \users in the path for the complete \payment operation, 
which demonstrates the practicality of \sysname.

\pparagraph{Scalability} In order to test the scalability of the \pay operation in \sysname, we have studied the 
running time and communication overhead required by each of the roles in a \payment (i.e., sender, 
receiver, and intermediate \user). Here, we take into account that Sphinx requires to pad the forwarded messages to the maximum 
path length. In the absence of widespread \pcn in practice, we set the maximum path length to $10$ in our test, as suggested for 
similar payment networks such as the Ripple credit network~\cite{silentwhispers}.

Regarding the computation time, the sender requires $3.09$ seconds to create $\pi_i$ for each intermediate \user. 
However, this computation time can be improved if different $\pi_i$ are calculated in parallel taking advantage of 
current multi-core systems. Each intermediate \user requires $130$ ms as only has to check the contract for  
\pcs with successor and predecessor \user in the path. Finally, the receiver incurs in few ms as she only has to 
check whether a given value is the correct pre-image of a given hash value.

Regarding communication overhead, the sender must create a message with $10$ proofs of knowledge and other 
few bytes associated to the contract for each intermediate \pc. So in total, the sender must forward $17$MB approximately. 
As Sphinx requires padded messages at each node to ensure anonymity, every intermediate \user must forward 
a message of the same size. 

In summary, these results show that even with an unoptimized implementation, 
a \payment with $10$ intermediate \users takes less than $5$ seconds 
and require a communication overhead of approximately $17$MB at each intermediate \user. Therefore, \sysname induces a relatively small overhead 
while enabling \payments between any two \users 
in the \pcn  
and  has the potential to be deployed as a  \pcn with 
a growing base of \users performing \payments with even $10$ intermediate \users  in a 
matter of few seconds, a result in line with other  privacy
preserving payment systems~\cite{privpay-ndss,
  silentwhispers}.

\pparagraph{Non-blocking \payments (\sysnamenonblocking)}
Given their similar definitions, the performance evaluation for \sysname carries over to \sysnamenonblocking. Additionally, 
the management of non-blocking \payments requires that intermediate \users maintain a list (\cur)  of current in-flight \payments 
and a queue (\q) of \payments waiting to be forwarded when capacity is available. The management of these data structures 
requires a fairly small computation overhead. Moreover, the number of messages to be stored in these data structures 
according to  the specification of \sysnamenonblocking 
is clearly linear in the length of the path. Specifically, a \payment involving a path of length $k\in \mathbb{N}$ incurs O(c$\cdot$k) message complexity, 
where $c$ is
 bounded by the total of concurrent conflicting \payments.


%
%
%

\section{Related Work}
\label{sec:related-work}

Payment channels were first introduced by the Bitcoin community~\cite{bitcoin-contract} and since then, several  extensions have been proposed. 
Decker and Wattenhofer~\cite{dmc} describe bidirectional \pcs~\cite{dmc}. Lind et al.~\cite{teechan} leverage 
trusted platform modules  to 
use a \pc without hindering compatibility with Bitcoin. However, these works focus on a single \pc and 
their extension to support \pcns remain an open challenge.

TumbleBit~\cite{tumblebit} and Bolt~\cite{bolt} propose \offchain path-based \payments  
while achieving sender/receiver anonymity in Tumblebit and payment 
anonymity in  Bolt. However, these approaches are restricted to 
single hop \payments, and it is not clear how to extend them  to account for generic multi-hop \pcns and provide the privacy notions of 
interest, as achieved by  \sysname and \sysnamenonblocking. 

\REMOVEPfor{SGAM}{170519}{
Flare~\cite{flare} is a novel routing mechanism to find paths between
any two \users in a \pcn. Flare uses an onion routing-like approach to
provide privacy guarantees in the \test operation. However, their
approach to calculate the capacity of a path 
leaks the intermediate values of the users in the path, thus
leaking sensitive information. }

The Lightning Network~\cite{ln} has emerged as the most prominent proposal for a \pcn in Bitcoin. Other \pcns such as 
Thunder~\cite{npc-impl4} and Eclair~\cite{npc-impl3} for Bitcoin and Raiden~\cite{raiden-nw} for Ethereum are being proposed 
as slight modifications of the Lightning Network. 
Nevertheless, their use of HTLC leaks a common identifier per \payment, thereby
 reducing the anonymity guarantees as we described in this 
work.  Moreover, current proposals lack a non-blocking solution for concurrent \payments. \sysname and \sysnamenonblocking, 
instead, rely on \newhtlc to overcome the linkability issue with HTLC. They provide a tradeoff between 
non-blocking progress and anonymity.

 Recent works~\cite{privpay-ndss,
  silentwhispers} propose privacy definitions for credit
networks, a payment system that supports multi-hop payments similar to \pcns. 
Moreover, privacy preserving protocols are described for both centralized~\cite{privpay-ndss}
and decentralized credit networks~\cite{silentwhispers}.
However, credit networks differ from \pcns in that they do not require
to ensure accountability against an underlying blockchain. 
This requirement 
reduces the  set of cryptographic operations available to design 
a \pcn. Nevertheless, 
\sysname and \sysnamenonblocking provide similar privacy guarantees as 
credit networks even  
under those restrictions.

\REMOVEPfor{SGAM}{170519}{ Moreover,
the centralized architecture of PrivPay conflicts with the decentralized
nature of a \pcn. SilentWhispers, despite being a decentralized solution, relies on
a pre-defined set of well connected and reputed nodes (called
landmarks) in the network that act as rendezvous point for each
transaction. However, these are not necessarily available in \pcns
hindering the portability of similar cryptographic techniques to \sysname.  }

\NEWPfor{SGAM}{170519}{
Miller et al~\cite{sprites} propose a construction for \pcs to reduce the time 
that funds are locked at intermediate \pcs (i.e., collateral cost), an interesting problem but 
orthogonal to our work. Moreover, they formalize their construction for multi-hop 
\payments as an ideal functionality. However, they focus on collateral cost and 
do not discuss privacy guarantees, concurrent \payments are handled in a blocking manner only, 
and their construction relies on smart contracts available on Ethereum that are 
incompatible with the current Bitcoin scripting system.
}

\NEWPfor{SGAM}{170519}{
Towns proposed~\cite{htlc-snarks} a variation of the HTLC contract, 
based on zk-SNARKs, to avoid its linkability problem among \pcs in a path. However, 
the Bitcoin community has not adopted this approach due to its inefficiency.  
In this work, we revisit this solution with a formal protocol with provable security and 
give an efficient instantiation based on ZK-Boo~\cite{zkboo}.}
\TODOP{Add sth similar for dlog contract if added in main body finally}
\TODOM{this sentence is dangerous, since it impacts the novelty of our work: we should also say that the Bitcoin community didn't manage to come up with an efficient solution (they used SNARKS),  which is on of our contributions. And this should also be mentioned in the intro, no?}

\section{Conclusion}
\label{sec:conclusions}

Permisionless blockchains governed on global consensus protocols face, among others, 
scalability issues in catering a growing base of \users and \payments. A burgeoning approach 
to overcome this challenge consists of \pcns and recent efforts have derived in the 
first yet alpha implementations such as the Lightning Network~\cite{ln} in Bitcoin or Raiden~\cite{raiden-nw} in Ethereum. 
We are, however, only scratching the surface as many challenges such as liquidity, network formation, routing 
scalability, concurrency or privacy are yet to be thoroughly studied. 

In this work, we lay the foundations for privacy and concurrency in \pcns. In particular, we formally 
define in the Universal Composability framework two modes of operation for \pcns attending to 
how concurrent \payments are handled (blocking versus non-blocking).
We provide formally proven 
instantiations (\sysname and \sysnamenonblocking) for each, offering a tradeoff between 
non-blocking progress and anonymity.
Our evaluation results demonstrate that is feasible to deploy 
\sysname and \sysnamenonblocking in practice and can scale to cater a growing number of \users.

\REMOVEPfor{SGAM}{170613}{Although in this work we focus on cryptocurrency-backed \pcn, we believe that our findings in this work 
are of interest in other decentralized payment networks such as credit networks~\cite{silentwhispers} or 
the InterLedger Protocol~\cite{ilp}, designed to allow multi-hop \payments across different blockchains. }

\pparagraph{Acknowledgments} We thank the anonymous reviewers for their helpful reviews, and Ivan Pryvalov for providing 
his python-based implementation of ZK-Boo.

This work is partially supported by a Intel/CERIAS research assistantship,
and by the National Science Foundation under grant CNS-1719196. This research is based upon work supported by the German research foundation (DFG) through the collaborative research center 1223 and by the state of Bavaria at the Nuremberg Campus of Technology (NCT). NCT is a research cooperation between the Friedrich-Alexander-University Erlangen-N\"urnberg (FAU) and the Technische Hochschule N\"urnberg Georg Simon Ohm (THN).

\bibliographystyle{abbrv}
\bibliography{bibliography,references}
\appendix
%
%
%

\section{Security Analysis}
\label{sec:uc-security-proof}

Our proof strategy consists of the description of a simulator $\simulate$ that handles users corrupted by the attacker and simulates the real world execution protocol while interacting with the ideal functionality $\ttp$. The simulator $\simulate$ spawns honest users at adversarial will and impersonates them until the environment $\environment$ makes a corruption query on one of the users: At this point $\simulate$ hands over to $\adv$ the internal state of the target user and routes all of the subsequent communications to $\adv$, who can reply arbitrarily. For operations exclusively among corrupted users, the environment does not expect any interaction with the simulator. Similarly, communications exclusively among honest nodes happen through secure channels and therefore the attacker does not gather any additional information other than \NEWPfor{SGAM}{170824}{the fact that} the communication took place. For simplicity, we omit these operations in the description of our simulator. The random oracle $H$ is simulated by $\simulate$ via lazy-sampling. The operations to be simulated for a \pcn are described in the following.

\REMOVEG{Note that there is a slight discrepancy between the notion of time in the ideal wold and in the real world protocol: In the former case the elapsing of time is modeled as the increasing size of the blockchain, while in the latter it is a global discrete counter. \TODOG{@Pedro: I think this is not accurate. In the real protocol it is also counted as the increased size in the blockchain. We just used days in the description for simplicity. How to have a global discrete counter in a decentralized network?}It is easy to see that the simulator can mimic the latter in the ideal world by regularly performing dummy transactions between honest nodes (since the blockchain is modeled as an append-only bulletin board). Throughout the following analysis we omit such a technicality in favor of a more understandable proof.}

\medskip
\noindent $\openc(\pair{\uid_1}{\uid_2}, \beta, t, f)$: Let $\uid_1$ be the user that initiates the request. We analyze two possible cases: 
		\begin{enumerate}
		\item {\em Corrupted $\uid_1$:} $\simulate$ receives a $(\pair{\uid_1}{\uid_2}, \beta, t, f)$ request from the adversary on behalf of $\uid_1$ and initiates a two-user agreement protocol with $\adv$ to convey upon a local fresh channel identifier $\pair{\uid_1}{\uid_2}$. If the protocol successfully terminates, $\simulate$ sends $(\openchannel, \pair{\uid_1}{\uid_2}, \beta, t, f)$ to $\ttp$, which eventually returns $(\pair{\uid_1}{\uid_2}, h)$.
		\item {\em Corrupted $\uid_2$:} $\simulate$ receives a message $(\pair{\uid_1}{\uid_2}, \val, t, f)$ from $\ttp$ engages $\adv$ in a two-user agreement protocol on behalf of $\uid_1$ for the opening of the channel. If the execution is successful, $\simulate$ sends an accepting message to $\ttp$ which returns $(\pair{\uid_1}{\uid_2}, h)$, otherwise it outputs $\bot$.
		\end{enumerate}	
	If the opening was successful the simulator initializes an empty list $\mathcal{L}_{\pair{\uid_1}{\uid_2}}$ and appends the value $(h, \val, \bot, \bot)$.
	
\medskip	
\noindent $\closec(\pair{\uid_1}{\uid_2}, \val)$: Let $\uid_1$ be the user that initiates the request. We distinguish two possible scenarios: 
	\begin{enumerate}
		\item {\em Corrupted $\uid_1$:} $\simulate$ receives a closing request from the adversary on behalf of $\uid_1$, then it fetches $\mathcal{L}_{\pair{\uid_1}{\uid_2}}$ for some value $(h, \val, x, y)$. If such a value does not exist then it aborts. Otherwise it sends $(\closechannel, \pair{\uid_1}{\uid_2}, h)$ to $\ttp$. 
		\item {\em Corrupted $\uid_2$:} $\simulate$ receives $(\pair{\uid_1}{\uid_2}, h, \bot)$ from $\ttp$ and simply notifies $\adv$ of the closing of the channel $\pair{\uid_1}{\uid_2}$. 
	\end{enumerate}
	
\medskip	
\noindent $\pay((\pair{\uid_0}{\uid_1}, \dots, \pair{\uid_n}{\uid_{n+1}}), \val)$: Since the specifications of the protocol differ depending on whether a user is a sender, a receiver or an intermediate node of a payment, we consider the cases separately.
	\begin{enumerate}
		\item {\em Sender:} In order to initiate a payment, the adversary must provide each honest user $\uid_i$ involved with a message $m_i$ that the simulator parses as $(\pair{\uid_{i-1}}{\uid_i},\pair{\uid_{i}}{\uid_{i-1}}, x_i, y_i, y_{i+1}, \allowbreak \pi_i, \val_i, t_i, t_{i+1})$, also the receiver of the payment $\uid_{n+1}$ (in case it is not corrupted) is notified with some message $(\pair{\uid_{n}}{\uid_{n+1}}, \allowbreak x_n, \allowbreak y_n, \val, t_{n})$. For each intermediate honest user $\uid_i$, the simulator checks whether $t_i \ge t_{i+1}$ and $\verifier((H, \allowbreak y_i, y_{i+1}, x_i), \pi_i) = 1$. If the conditions hold, $\simulate$	 sends to $\ttp$ the tuple $(\pay, \val_i, \allowbreak (\pair{\uid_{i-1}}{\uid_i}, \allowbreak \pair{\uid_i}{\uid_{i+1}}), t_{i-1}, t_i)$, whereas for the receiver (in case it is honest) sends $(\pay, \allowbreak \val, \pair{\uid_{n}}{\uid_{n+1}}, t_n)$ if $y_n = H(x_n)$\REMOVEG{. }\NEWG{, otherwise it aborts.} For each intermediate user $\uid_i$ the simulator confirms the payment only when receives from the user $\uid_{i+1}$ an $x$ such that $H(x_i \oplus x) = y_i$. If $\adv$ outputs a value $x^*$ such that $H(x^*) = y_{i+1}$ but $H(x_i \oplus x^*) \ne y_i$ then $\simulate$ aborts the simulation. If the receiver is honest then the simulator confirms the payment if the amount $\val$ corresponds to what agreed with the sender and if $H(x_n) = y_n$. If the payment is confirmed the entry $(h_i, \val^* - \val_i, x_i \oplus x, y_i)$ is added to $\mathcal{L}_{\pair{\uid_{i-1}}{\uid_i}}$, where $(h_i^*, \val^*, \cdot, \cdot)$ is the entry of $\mathcal{L}_{\pair{\uid_{i-1}}{\uid_i}}$ with the lowest $\val^*$, and the same happens for the receiver.
		\item {\em Receiver:} $\simulate$ receives some $(h, \pair{\uid_n}{\uid_{n+1}}, \val, t_n)$ from $\ttp$, then it samples a random $x \in \bit{\lambda}$ and returns to $\adv$ the tuple $(x, H(x), \val)$. If $\adv$ returns a string $x ' = x$, then $\simulate$ returns $\top$ to $\ttp$, otherwise it sends $\bot$.
		\item {\em Intermediate user:} $\simulate$ is notified that a corrupted user is involved in a payment with a message of the form $(h_i, h_{i+1}, \allowbreak \pair{\uid_{i-1}}{\uid_i}, \pair{\uid_{i}}{\uid_i+1}, \val, t_{i-1}, t_i)$ by $\ttp$. $\simulate$ samples an $x \in \bit{\lambda}$ and an $x' \in \bit{\lambda}$ and runs the simulator of the zero-knowledge scheme to obtain the proof $\pi$ over the statement $(H, H(x \oplus x '), H(x'), x)$. The adversary is provided with the tuple $(\pair{\uid_{i-1}}{\uid_i},\pair{\uid_{i}}{\uid_{i-1}}, x,H(x \oplus x'), H(x'), \pi, \val, t_{i-1}, t_i)$ via an anonymous channel. If $\adv$ outputs a string $x'' = x \oplus x '$, then $\simulate$ aborts the simulation. At some point of the execution the simulator is queried again on $(h_i, h_{i+1})$, then it sends $x'$ to $\adv$ on behalf of $\uid_{i+1}$. If $\adv$ outputs a string $z = x \oplus x'$ the simulator sends $\top$ to $\ttp$ and appends $(h_i, \val^*- \val, z, H(z))$ to $\mathcal{L}_{\pair{\uid_{i-1}}{\uid_i}}$, where $(h_i^*, \val^*, \cdot, \cdot)$ is the entry of $\mathcal{L}_{\pair{\uid_{i-1}}{\uid_i}}$ with the lowest $\val^*$. The simulator sends $\bot$ otherwise. Note that we consider the simpler case where a single node in the payment is corrupted. However this can be easily extended to the more generic case by book-keeping the values of $h_i$ and choosing the corresponding the pre-images $x$ and $x'$ consistently. The rest of the simulation is unchanged.
	\end{enumerate}

\pparagraph{Analysis} Since the simulation runs only polynomially-bounded algorithms it is easy to see that the simulation is efficient. We now argue that the view of the environment in the simulation is indistinguishable from the execution of the real-world protocol. For the moment, we assume that the simulation never aborts, then we separately argue that the probability of the simulator to abort is negligible. For the $\openc$ and $\closec$ algorithms the indistinguishability argument is trivial. On the other hand for the payment we need a more sophisticated reasoning. Consider first the scenario where the sender is corrupted: In this case the simulation diverges form the the original protocol since each multi-hop payment is broke down into separate single-hop payments. Note that the off-chain communication mimics exactly the real-world protocol (as long as $\simulate$ does not abort): Each node $\uid_i$ that is not the receiver confirms the transaction to $\ttp$ only if it learns a valid pre-image of its $y_i$. Since we assume that the simulation does not abort, it follows that the simulation is consistent with the fact that each honest node always returns $\top$ at this stage of the execution (i.e., the payment chain does not stop at a honest node, other than the sender). However, the values published in the blockchain could in principle diverge from what the adversary is expecting in the real execution. In fact, an entry of the real blockchain contains the values of $(x, H(x))$ corresponding to a particular payment, in addition to the information that is leaked by the ideal functionality. Therefore we have to show that the values of $(x, y)$ that the simulator appends to $\adv$'s view of $\blockchain$ (in the $\closec$ simulation) have the same distribution as in the real world. Note that those values are either selected by the adversary if the sender is corrupted (there the argument is trivial) or chosen to be $(x, H(x))$ by the simulator, for some randomly chosen $x \in \{0,1\}^{\lambda}$. For the latter case it is enough to observe that the following distributions are statistically close 
\[
\left(\left(\bigoplus^n_{i=1} x_i, y_1\right), \dots, (x_n, y_n)\right) \approx ((r_1, s_1), \dots, (r_n, s_n)),
\]
where for all $i: (x_i, r_i) \gets \{0,1\}^{2\cdot \lambda}$, $y_i \gets H(x_i)$, and $s_i \gets H(r_i)$. Note that on the left hand side of the equation the values are distributed accordingly to the real-world protocol, while on the right hand side the distribution corresponds to the simulated values. The indistinguishability follows. For the simulation of the receiver and of the intermediate users one can use a similar argument. We only need to make sure that $\adv$ cannot interrupt a payment chain before it reaches the receiver, which is not allowed in the ideal world. It is easy to see that in that case ($\adv$ outputs $x''$ such that $H(x'') = H(x \oplus x')$ before receiving $x'$) the simulation aborts.

What is left to be shown is that the simulation aborts with at most negligible probability. Let $\mathsf{abort}_\text{s}$ the event that $\simulate$ aborts in the simulation of the sender and let $\mathsf{abort}_\text{i}$ be the event that $\simulate$ aborts in the simulation of the intermediate user. By the union bound we have that
$\prob{\mathsf{abort}} \le \prob{\mathsf{abort}_\text{s}} + \prob{\mathsf{abort}_\text{i}}$.


We note that in case $\mathsf{abort}_\text{s}$ happens than the adversary was able to output a valid proof $\pi_i$ over $(H, y_i, y_{i+1}, x_i)$ and an $x^*$ such that $H(x^*) = y_{i+1}$ and $H(x^* \oplus x_i) \ne y_i$. Let $w$ be a bitstring such that $H(w) = y_{i+1}$ and $H(w \oplus x_i) = y_i$, by the soundness of the proof $\pi_i$ such a string is guaranteed to exists. It follows that $H(x^* \oplus x_i) \ne H(w \oplus x_i)$ which implies that $w \ne x^*$, since $H$ is a deterministic function. However we have that $H(x^*) = H(w)$, which implies that $w = x^*$, since $\adv$ can query the random oracle at most polynomially-many times. This is a contradiction and therefore it must be the case that for all $\ppt$ adversaries the probability of $\mathsf{abort}_\text{s}$ to happen is $0$. We can now rewrite $\prob{\mathsf{abort}} \le \prob{\mathsf{abort}_\text{i}}$.
Consider the event $\mathsf{abort}_\text{i}$: In this case we have that $\adv$, on input $(H(x \oplus x'), H(x'), x)$,  is able to output some $x'' = x \oplus x'$. Note that $x'$ is a freshly sampled value and therefore the values $H(x \oplus x')$ and $H(x')$ are uniformly distributed over the range of $H$. Thus the probability that $\adv$ is able to output the pre-image of $H(x \oplus x')$ without knowing $x'$ is bounded by a negligible function in the security parameter. It follows that $\prob{\mathsf{abort}} \le \negl(\lambda)$.
And this concludes our proof. \qed

\pparagraph{Non-Blocking Solution} The security proof for our non-blocking solution is identical to what described above, with the only exception that the ideal functionality leaks the identifier of a payment to the intermediate users. Therefore the simulator must make sure to choose the transaction identifier consistently for all \NEWG{of} the corrupted users involved in the same payment. \NEWG{In addition to that, the simulator must also implement the non-blocking logic for the queueing of the payments.} The rest of the argument is unchanged.

%
%
%
%

\section{Agreement between Two \Users}
\label{sec:agreement-two-users}

\REPLACEPfor{SGAM}{170519}{A fundamental building block towards realizing a \pcn that handles concurrent \payments in a non-blocking 
manner involves reaching \emph{agreement}~\cite{FLP85} among the
two \users sharing a \pc on its state at each point in time. Given that, in the following we define the agreement between two \users. Later, we describe the operations 
required by a \pcn. }
{
In this section, we describe the protocol run by two \users, $\uid_0$ and $\uid_1$, sharing a \pc to reach \emph{agreement}~\cite{FLP85} 
on the channel's state at each point in time. 
}

\TODOP{Probably not really important now, but this is confusing the way 
it is explained now}

\pparagraph{Notation and Assumptions}
In this section, we follow the notation we introduced in~\cref{sec:construction}. 
We assume that there is a total order between the events received 
by the \users at a \pc (e.g., 
lexicographically sorted by the hash of the corresponding \payment data) and the \users 
(e.g., lexicographically sorted by their public verification keys). Moreover, 
we assume that \users perform the operations associated to each event 
as defined in our construction (see~\cref{sec:protocol-details}). Therefore, 
in this section we only describe the additional steps required by \users to handle 
concurrent \payments. Finally, we assume that two \users sharing a \pc,    
locally maintain the state of the \pc (\pcstate). The actual definition of \pcstate 
depends on whether concurrent \payments are handled in a blocking or non-blocking 
manner. For blocking, \pcstate is defined as  $\capa(\pair{\uid_0} {\uid_{1}})$, where $\capa$ 
denotes the current capacity in the \pc. For non-blocking, \pcstate is 
defined as a tuple $\{\cur[\,], \q[\,], \valtext\}$, where $\cur$ denotes an array of \payments 
 currently using (part of) the \capacity available  at the \pc;  
$\q$ denotes the array of \payments waiting for enough \capacity at the \pc.

The agreement on \pcstate between the corresponding two \users $\uid_0$ and $\uid_1$ is performed in two communication rounds. In the first  round,  
both \users exchange the set of events $\{\ops_b\}$ to  be applied into the \pcstate. 
At the end of this first round, each \user comes up with the aggregated set of events $\{\ops\} \df \{\ops\}_0 \cup \{\ops\}_1$ deterministically sorted 
according to the following criteria. First, the events proposed by the \user with the highest identifier are included first. 
Second, if several events are included 
in $\{\ops\}_b$, they are sorted  according to the following sequence: $\confirm, \release, \hold$.\footnote{Although other sequences are possible, 
we fix this one to ensure that the sorting is deterministic.} 
Finally, events of the same type 
are sorted in decreasing order by the corresponding \payment identifier. These set of rules ensure that the both 
\users can deterministically compute the same sorted version of the set $\{\ops_i\}$. 

Before starting the second communication round, each \user applies the changes related to each event in $\{\ops_i\}$  to the current \pcstate. 
The mapping between each event and the corresponding actions is defined as  
a function $\set{(\ops_j, m_j)} $ $\gets \agreef(\{\ops_i\})$. This function returns a set of tuples that indicate what events must be forwarded to which \user in the \payment path.
Then, in the second communication round, each event $\ops_j$ is sent to the corresponding \user $\uid_j$ 
(encoded in $m_j$). 
The actual implementation of the function \agreef determines how the concurrent \payments are handled. 
In \sysname, we implement the function \agreef 
as described in~\cref{fig:pay-operations,fig:pay-operations-two} (black pseudocode) for blocking approach and 
as described in~\cref{fig:pay-operations,fig:pay-operations-two} (\colorchange pseudocode) for non-blocking approach. 

In the following, we denote the complete agreement protocol between two 
\users by $\consensustwo(\uid_0, \uid_1,$ $ \set{\ops_i})$.

\begin{lemma}
\label{sec:lemma-2users}
$\consensustwo(\uid_0, \uid_1, \set{\ops_i})$ ensures agreement on the \pcstate given the 
set of events $\set{\ops}$.
\end{lemma}
\begin{proof}
Assume that \pcstate is consistent between two \users $\uid_i$ and $\uid_j$ before
 $\consensustwo(\uid_0, \uid_1, \set{\ops_i})$ is invoked. It is easy to see that both \users 
 come with the same sorted version of  $\set{\ops_i}$ since the sorting rules are deterministic.  
  Moreover, for each event, the function $\agreef$ deterministically updates \pcstate 
  and returns a tuple $(m, \ops)$. 
  As the events are applied in the same order by both \users, 
 they reach agreement on the same 
  updated \pcstate and the same set of tuples $\set{(\uid_k, \ops_k)}$.
  \end{proof}

%
\iftechreport
\begin{figure*}
\else
\begin{figure*}[tb]
\fi
\begin{mdframed}
{\bf Open channel:} On input $(\openchannel, \pair{\uid}{\uid'}, \val, \uid', \timeout, \fee)$ from a user $\uid$, $\ttp$ checks whether $\pair{\uid}{\uid'}$ is well-formed (contains valid identifiers and it is not a duplicate) and eventually sends $(\pair{\uid}{\uid'}, \val, \timeout, \fee)$ to $\uid'$, who can either abort or authorize the operation. In the latter case, $\ttp$ appends the tuple $(\pair{\uid}{\uid'}, \val, \timeout, \fee)$ to $\blockchain$ and the tuple $(\pair{\uid}{\uid'}, \val, \timeout, \cstate)$ to $\mathcal{L}$, for some random $\cstate$. $\ttp$ returns $\cstate$ to $\uid$ and $\uid'$.

\smallskip
{\bf Close channel:} On input $(\closechannel, \pair{\uid}{\uid'}, \cstate)$ from a user $\in \{ \uid', \uid \}$ the ideal functionality $\ttp$ parses $\blockchain$ for an entry $(\pair{\uid}{\uid'}, \val, \timeout, \fee)$ and $\mathcal{L}$ for an entry $(\pair{\uid}{\uid'}, \val', \timeout', \cstate)$, for $\cstate \ne \bot$. 
  If $\pair{\uid}{\uid'} \in \mathcal{C}$ or \REMOVEGfor{SAM}{170519}{$\pair{\uid}{\uid'}=\bot$} or $\timeout > |\blockchain|$ or $\timeout' > |\blockchain|$, 
  the functionality aborts. 
	Otherwise,  $\ttp$ adds the entry $(\pair{\uid}{\uid'}, \val', \timeout', \fee)$ to $\blockchain$ and adds 
	$\pair{\uid}{\uid'}$ to $\mathcal{C}$. $\ttp$ then notifies both users involved with a message $(\pair{\uid}{\uid'}, \bot, \cstate)$. \TODOP{Why $\bot$ here?}

\smallskip
{\bf Payment:} On input $(\pay, \val, 
	(\pair{\uid_0}{\uid_1}, \dots, \pair{\uid_n}{\uid_{n+1}}), (\timeout_0, \dots, \timeout_n), \hl{\Txid})$ from a user $\uid_0$, $\ttp$ executes the following interactive protocol: 
	\begin{enumerate}
	\item  For all $i \in \{1, \dots, (n+1)\}$, $\ttp$ parses $\blockchain$ for an entry of the form 
		$((\pair{\uid_{i-1}}{\uid'_{i}}, \val_i, \timeout_i', \fee_i))$. If such an entry does exist, $\ttp$ sends the tuple $(\hl{\Txid}, \hl{\Txid},$ 
		$\pair{\uid_{i-1}}{\uid_i}, \pair{\uid_i}{\uid_{i+1}}, \val - \sum^n_{j=i}\fee_j, \timeout_{i-1}, \timeout_i)$ to the user $\uid_i$ via an anonymous channel 
		(for the specific case of the receiver the tuple is only $(\hl{\Txid}, \pair{\uid_{n}}{\uid_{n+1}} , \val, \timeout_n)$). Then, $\ttp$ checks whether 
		for all entries of the form 
		 $(\pair{\uid_{i-1}}{\uid_{i}}, \val_i', \cdot, \cdot) \in \mathcal{L}$ it holds that $\val_i' \ge \left(\val - \sum^n_{j=i}f_j\right)$ and that $\timeout_{i-1} \ge \timeout_i$. If this is the case, $\ttp$ adds $d_i = (\pair{\uid_{i-1}}{\uid_{i}}, \val_i' - (\val - \sum^n_{j=i}\fee_j), \timeout_i, \bot)$  to $\mathcal{L}$, where $(\pair{\uid_{i-1}}{\uid_{i}}, \val_i', \cdot , \cdot) \in \mathcal{L}$ is the entry with the \REPLACEGfor{SAM}{170519}{highest}{lowest} $\val_i'$ \hl{and sets $\queued = n+1$}.
		 Otherwise, $\ttp$ performs the following steps:
			\begin{itemize}
			\item \hl{If there exists an entry of the form $(\pair{\uid_k}{\uid_{k+1}}, -, -, \Txid^*) \in \mathcal{L}$ such that $\Txid > \Txid^*$, then $\ttp$ adds 
			 $d_l = (\pair{\uid_{l-1}}{\uid_{l}}, \val_l' - (\val + \sum^n_{j=l}\fee_j), \timeout_l, \bot)$  to $\mathcal{L}$,
			for $l \in \set{1,\dots,k}$ . Additionally, $\ttp$ adds $(\Txid, (\pair{\uid_k}{\uid_{k+1}}, \dots,  \pair{\uid_n}{\uid_{n+1}}), \val - \sum^n_{j=k}\fee_j, (t_k, \dots, t_n))  \in \queue$. Finally, $\ttp$ sets $\queued = k$. }
			\item Otherwise, $\ttp$ removes from $\mathcal{L}$ all the entries $d_i$ added in this phase. \hl{Additionally, $\ttp$ looks for entries of the 
			form $ (\Txid', (\pair{i}{i+1}, \dots, \pair{\tilde{n}}{\tilde{n}+1}), \tilde{\val}, (\timeout_i, \dots, \tilde{\timeout_n}) ) \in \queue$, deletes them and 
			execute $(\pay, \tilde{\val}, (\pair{i}{i+1}, \dots, \pair{\tilde{n}}{\tilde{n}+1}),  (\timeout_i, \dots, \tilde{\timeout_n})).$}
			\end{itemize}

		 \item For all $i \in \{\hl{\queued}, \dots, 1\}$ $\ttp$ queries all $\uid_i$ with $(\cstate_i, \cstate_{i+1})$, through an anonymous channel. Each user can reply with either $\top$ or $\bot$. Let $j$ be the index of the user that returns $\bot$ such that for all $i > j: \uid_i$ returned  $\top$. If no user returned $\bot$ we set $j = 0$.

		\item For all $i \in \{ j+1, \dots, \hl{\queued}\}$ the ideal functionality $\ttp$ updates $d_i \in \mathcal{L}$ (defined as above) to $(-, -, -, \Txid)$ and notifies the user of the success of the operation with with some distinguished message 
		$(\success, \Txid, \Txid)$. For all $i \in \{0, \dots, j\}$ (if $j \ne 0$) $\ttp$ performs the following steps: 
		\begin{itemize}
		\item Removes $d_i$ from $\mathcal{L}$ and notifies the user with the message $(\bot, \hl{\Txid}, \hl{\Txid})$.
		\item \hl{$\ttp$ looks for entries of the 
			form $ (\Txid', (\pair{i}{i+1}, \dots, \pair{\tilde{n}}{\tilde{n}+1}), \tilde{\val}, (\timeout_i, \dots, \tilde{\timeout_n}) ) \in \queue$, removes them from $\queue$ 
			and  execute $(\pay, \tilde{\val}, (\pair{i}{i+1}, \dots, \pair{\tilde{n}}{\tilde{n}+1}),  (\timeout_i, \dots, \tilde{\timeout_n})).$}
		\end{itemize}		
	\end{enumerate}
\end{mdframed}
\vspace{-1em}
\caption{Ideal world functionality for \pcns for non-blocking progress.\label{fig:ideal-world-non-blocking}}
\vspace{-1em}
\end{figure*}

%

\section{Ideal World Functionality for Non-Blocking Payments}
\label{sec:ideal-world-non-blocking}

In this section, we detail the ideal world functionality for a \pcn that handles concurrent 
\payments in a non-blocking manner. We highlight in \colorchange the changes with 
respect to the ideal world functionality presented in~\cref{sec:ideal-world} that correspond 
to a \pcn that handles concurrent \payments in a blocking manner. Moreover, 
we assume the same 
model, perform the same assumptions and use the same notation as described in~\cref{sec:ideal-world}.
Additionally, we use the variable $\queued$ to track at which intermediate \user the 
\payment is queued if there is not enough capacity in her channel and the \payment 
identifier is higher than those in-flight. Moreover, we use a list $\queue$ to keep track 
of remaining hops for queued \payments. Entries in $\queue$ are of the form 
$((\pair{\uid_1}{\uid_2}, \dots, \pair{\uid_k}{\uid_{k+1}}), \val, (\timeout_1, \dots, \timeout_k))$ and 
contain the remaining list of \pcs $(\pair{\uid_1}{\uid_2}, \dots, \allowbreak\pair{\uid_k}{\uid_{k+1}})$, their 
associated timeouts $(\timeout_1, \dots, \timeout_k)$ and the 
remaining \payment value $\val$.

For simplicity we only model unidirectional channels, 
although our functionality can be easily extended to support also bidirectional channels. 
The execution of our simulation starts with $\ttp$ querying $\ttp_\blockchain$ to 
initialize it and $\ttp$ initializing itself the locally stored empty lists $\mathcal{L}, \mathcal{C}, \queue$.

%

\begin{figure}[h!]
\includegraphics[width=\columnwidth]{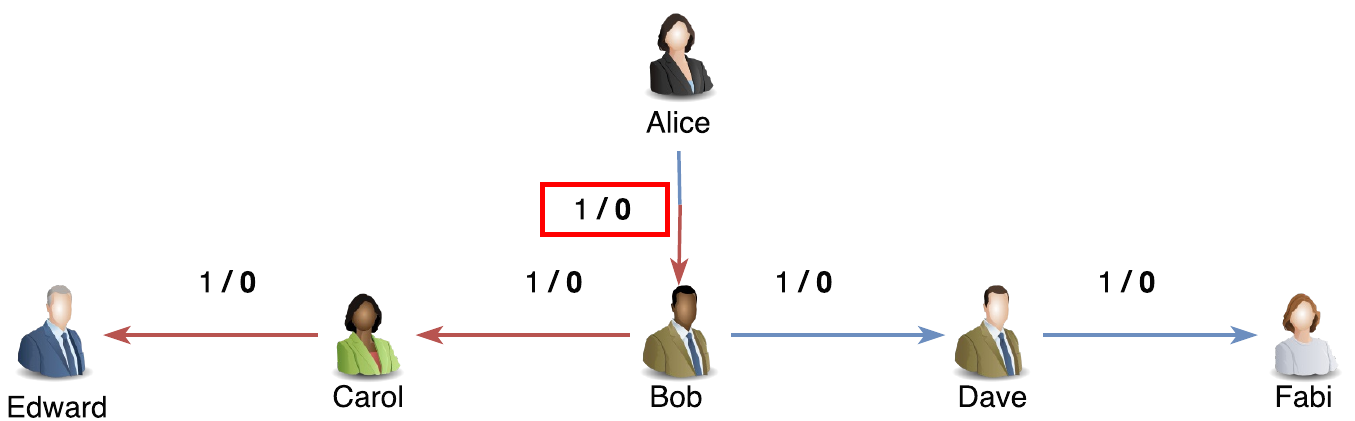}
\caption{ Execution depicting two \payments:\payment $\Txid_i$ from Alice to Edward and \payment $\Txid_j$
from Alice to Fabi. If Alice and Bob are byzantine, they can allow both \payments to be successful (while losing \REPLACEPfor{SGAM}{170824}{pay}{funds} themselves). \label{fig:bottleneck}}
\vspace{-1em}
\end{figure}


\section{Proof for Concurrency Lemmas}
\label{sec:concurrency-proofs-lemmas}

\iftechreport
\begin{figure}[H]
\else
\begin{figure}[b]
\fi
\vspace{-1em}
\includegraphics[width=\columnwidth]{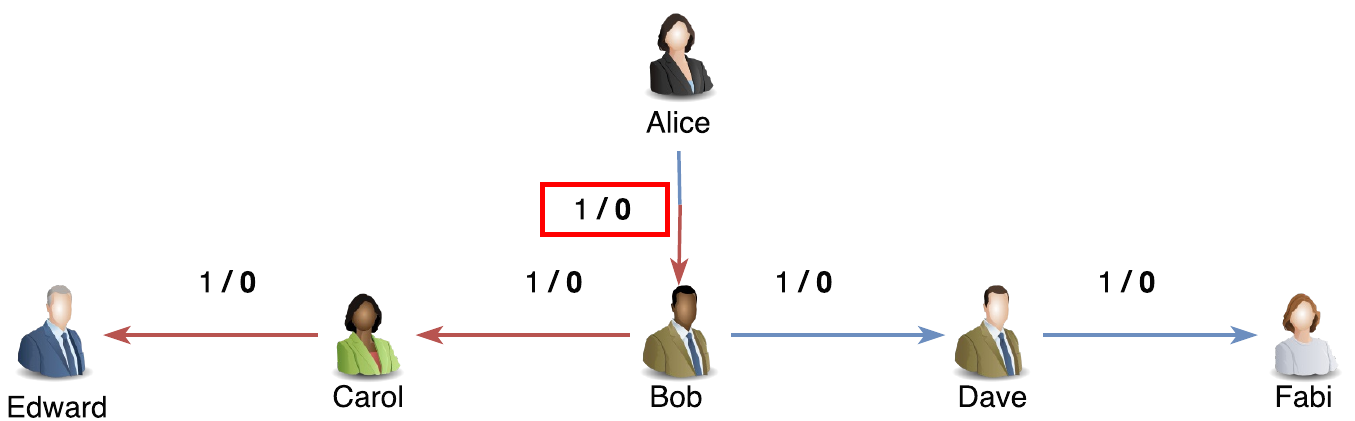}
\caption{ Execution depicting two \payments:\payment $\Txid_i$ from Alice to Edward and \payment $\Txid_j$
from Alice to Fabi. If Alice and Bob are byzantine, they can allow both \payments to be successful (while losing \REPLACEPfor{SGAM}{170824}{pay}{funds} themselves). \label{fig:bottleneck}}
\end{figure}


\begin{proof}[Proof for Lemma~\ref{lemma:byzantine}]
Consider an execution of two \payments depicted in Figure~\ref{fig:bottleneck}: \payment $\Txid_i$ from Alice to Edward and \payment $\Txid_j$
from Alice to Fabi. The \pc between Alice and Bob is a contending bottleneck for both $\Txid_i$ and $\Txid_j$, however, only one of the \payments can be
successfully executed since the \pc between Alice and Bob has the capacity for only one of the two to be successful. Suppose by contradiction that both
$\Txid_i$ and $\Txid_j$ are successfully completed. Indeed, this is possible since byzantine \users Alice and Bob can respond with an incorrect
\pc capacity to \users Edward and Fabi. 
However, 
the \pc between Alice and Bob does not have sufficient capacity for both transactions to be successful---contradiction since there does not exist any equivalence to the sequential specification of payments channels.
\end{proof}

%
\begin{proof}[Proof for Lemma~\ref{lemma:transaction-id}]
Suppose by contradiction that there exists a strictly serializable disjoint-access implementation providing non-blocking progress.
Consider the following payment network: 
$\uid_1$ $\rightarrow$ $\uid_2$ $\rightarrow$ $\uid_3$ $\rightarrow$ $\uid_4$ $\rightarrow$ $\uid_5$ $\rightarrow$ $\uid_1$.
Consider two concurrent \pay operations of the form $\pay_1(\pair{\uid_1}{\uid_2},$  $\pair{\uid_2}{\uid_3},$ $\pair{\uid_3}{\uid_4}, \pair{\uid_4}{\uid_5},\val)$ and 
$\pay_2(\pair{\uid_4}{\uid_5}, \allowbreak \pair{\uid_5}{\uid_1}, \pair{\uid_1}{\uid_2},$ $\pair{\uid_2}{\uid_3},$ $\val)$. 
Consider the execution $E$ in which $\pay_1$ and $\pay_2$ run concurrently up to the following step: $\pay_1$ executes from $\uid_1 \rightarrow \ldots \uid_4$
and $pay_2$ executes from $\uid_4$ $\rightarrow$ $\uid_5$ $\rightarrow$ $\uid_1$.
Let $E_1$ (and resp. $E_2$) be the extensions of $E$ in which $\pay_1$ (and resp. $\pay_2$) terminates \emph{successfully} and $\pay_2$
(and resp. $\pay_1$) terminates \emph{unsuccessfully}.
By assumption of non-blocking progress, there exists such a finite extension of this execution in which both $\pay_1$ and $\pay_2$ must terminate (though they may not
be successful since this depends on the available channel capacity).

Since the implementation is disjoint-access parallel, execution $E_1$ is \emph{indistinguishable} to $(\uid_1, \ldots, \uid_5)$
(and resp. $(\uid_4, \ldots, \uid_3)$) from the execution
$\bar E$, an extension of $E$, in which only $\pay_1$ (and resp. $\pay_2$) is successful \emph{and} matches the sequential specification of \pcn.
Note that analogous arguments applies for the case of $E_2$.

However, $E_1$ (and resp. $E_2$) is not a correct execution since it lacks the \emph{all-or-nothing} semantics: 
only a proper subset of the channels from the execution $E$ involved in $\pay_1$ (and resp. $\pay_2$)
have their capacities decreased  by $v$ (and resp. $v'$). This is a contradiction to the assumption of strict serializability, thus completing the proof.  
\end{proof}

\end{document}